\documentclass[12pt,draftclsnofoot,onecolumn]{IEEEtran}

\usepackage[T1]{fontenc}
\usepackage[utf8]{inputenc}
\usepackage{amsmath,amssymb,amsfonts,amsthm,mathtools,bm,bbm}
\usepackage{array}
\usepackage{graphicx}
\usepackage{placeins}
\usepackage{cite}
\usepackage{microtype}
\usepackage{url}
\usepackage[colorlinks=true,linkcolor=blue,citecolor=blue,urlcolor=blue,
pdftitle={Covert Block-Activity Information Transmission over Thermal-Loss Bosonic Channels: Latency--Payload Limits and Finite-Key Achievability},
pdfauthor={Qipeng Qian and Yuntao Qian}]{hyperref}
\usepackage{orcidlink}

\newtheorem{assumption}{Assumption}
\newtheorem{lemma}{Lemma}
\newtheorem{proposition}{Proposition}
\newtheorem{theorem}{Theorem}
\newtheorem{corollary}{Corollary}
\newtheorem{remark}{Remark}

\newcommand{\E}{\mathbb E}
\newcommand{\Prb}{\mathbb P}
\newcommand{\Tr}{\operatorname{Tr}}
\newcommand{\supp}{\operatorname{supp}}
\newcommand{\CN}{\mathcal{CN}}
\newcommand{\NN}{\mathcal N}
\newcommand{\cC}{\mathcal C}
\newcommand{\cM}{\mathcal M}

\newcommand{\cG}{\mathcal G}

\newcommand{\wt}{\operatorname{wt}}
\newcommand{\db}{d_{\rm b}}

\newcommand{\1}{\mathbbm 1}
\newcommand{\norm}[1]{\left\lVert#1\right\rVert}
\newcommand{\ket}[1]{\lvert#1\rangle}
\newcommand{\bra}[1]{\langle#1\rvert}
\newcommand{\Disp}{\mathsf D}

\begin{document}

\title{Covert Block-Activity Information Transmission over Thermal-Loss Bosonic Channels: Latency--Payload Limits and Finite-Key Achievability}

\author{Qipeng Qian\,\orcidlink{0009-0003-2407-1292}, Yuntao Qian\,\orcidlink{0000-0002-7418-5891}%
\thanks{This work has been submitted to the IEEE for possible publication.
Copyright may be transferred without notice, after which this version may no longer be accessible.\\
Supported by the National Natural Science Foundation of China under Grant 62571475
(corresponding author: Qipeng Qian).}
\thanks{Qipeng Qian is with Supcon Technology, Hangzhou, China and the College of Artificial Intelligence, Zhejiang University, Hangzhou 310027, China (email: qianqipeng@supcon.com).}%
\thanks{Yuntao Qian is with the College of Artificial Intelligence, Zhejiang University, Hangzhou 310027, China (email: ytqian@zju.edu.cn).}%
}

\markboth{IEEE Transactions on Information Theory, Submitted for Review}%
{Qian and Qian: Covert Block-Activity Information Transmission}

\maketitle

\begin{abstract}
We study covert block-activity information transmission over a thermal-loss bosonic channel, where messages are encoded in weight-constrained activity patterns that Bob must recover through block-local decisions by prescribed deadlines. Shared circular Gaussian displacement modulation makes Willie's averaged lost-light state exactly thermal, with a strictly convex relative-entropy cost in signal energy, whereas a fixed Gaussian receiver at Bob yields a Gaussian mean-shift divergence linear in energy. This asymmetry produces an exact finite-block energy--information frontier and a detector-independent latency converse, matched in order by a block-reset cumulative-sum detector. For block covertness budget $\delta_b$ and error target $\epsilon_b$, the required active length scales as $\delta_b^{-1}\log^2(1/\epsilon_b)$ up to an explicit channel--receiver factor. Lifting this local law to communication yields the matching transmission limit $\log M=\Theta(\sqrt n/\log n)$ for the symmetric coordinatewise architecture under a fixed total covertness budget, maximal-message covertness, vanishing maximal error, and local deadlines. A relaxed full-horizon reference with the same modulation family and fixed measurement supports $\Theta(\sqrt n)$, showing that covertness and the selected measurement alone do not impose the extra logarithmic factor. Finally, public quadrature phase-shift keying codebooks selected by $O(\sqrt n)$ secret bits remove ideal continuous shared randomness without changing the payload order.
\end{abstract}

\begin{IEEEkeywords}
Bosonic channels, covert communication, block-activity coding, sequential detection,
transmission limits.
\end{IEEEkeywords}

\section{Introduction}

Covert communication asks how much information Alice can send to Bob while keeping Willie from reliably deciding whether communication took place. For thermal-loss bosonic channels, several models obey a square-root law: over $n$ channel uses, the covert payload is of order $\sqrt n$ \cite{BashNatComm2015,Bullock2019,Gagatsos_2020,Anderson2025AchievabilityCovertQuantum}. This paper studies a different operational question: how much covert information can be carried by the \emph{temporal block-activity pattern itself} when that pattern must be recovered under local deadlines? We divide the horizon into blocks, map each message to a weight-constrained binary activity vector, and require Bob to output the active-or-idle decision for each block by the corresponding deadline. The activity vector is therefore the codeword, rather than a synchronization variable that is estimated only to enable a separate message decoder.

Encoding information in temporal position or support is not new. Pulse-position modulation uses pulse location as a symbol, and covert pulse-position constructions exploit sparse positions to achieve information-theoretic covert limits in classical channels \cite{BlochGuha2017PPM,KadampotTahmasbiBloch2020PPM}. Recent work has also studied sparse signaling over thermal-loss bosonic channels, where rare non-innocent uses provide a practical route to square-root-law operation \cite{TanAndersonBullockBash2026}. The contribution here is not position modulation or sparsity by itself. Our operational constraint is different: every coordinate of an $N$-dimensional activity vector is a separately deadline-constrained decision, all active coordinates share one total quantum-relative-entropy budget, and the entire vector must satisfy maximal-message reliability. This local-decision requirement couples covertness allocation, detection reliability, and combinatorial payload, and it is this coupling that produces the transmission limits studied below.

To see the mechanism, suppose a message activates at most $K$ of $N$ blocks under a fixed total Willie budget $\delta$. Increasing $K$ divides the covertness budget among more active coordinates. At the same time, recovering the complete pattern forces the false-alarm and missed-detection probabilities of the individual block decisions to decrease with the number of coordinates. Thus the number of admissible message patterns cannot be increased independently of the resources required to resolve each coordinate reliably. The central technical step is to quantify this local-to-global coupling rather than to analyze a detector in isolation.

The one-block analysis is driven by a useful asymmetry between Willie and Bob. Shared circular Gaussian modulation makes Willie's averaged lost-light state thermal, with an exact relative-entropy cost $g(E)$ that is strictly convex in the signal energy $E$. Once Bob fixes a Gaussian measurement and conditions on the shared symbol, his observation is a Gaussian mean shift with relative entropy $\kappa E$. Willie's cost is therefore locally quadratic in weak signal energy, whereas Bob's accumulated information is linear. We exploit this structure to obtain an exact finite-block energy--information frontier and then convert that frontier into a detector-independent latency limit.

For one block with covertness budget $\delta_b$ and target false-alarm and miss probability $\epsilon_b\downarrow0$, the minimum active length satisfies
\begin{equation}
L^*(\epsilon_b,\delta_b)
=\Theta\!\left(
\frac{c_W}{\kappa^2\delta_b}
\log^2\frac1{\epsilon_b}
\right).
\label{eq:intro-local-limit}
\end{equation}
The converse applies to every detector whose decision is available by the block deadline, and a block-reset cumulative-sum rule attains the same order. Lifting this local limit to communication gives the main end-to-end result. Under the symmetric coordinatewise architecture, with $N=\Theta(K)$, covertness-budget splitting and family-wise reliability force $L=\Theta(K\log^2K)$ and hence $n=\Theta(K^2\log^2K)$, while the number of admissible activity patterns satisfies $\log M=\Theta(K)$. Consequently,
\begin{equation}
\log M=\Theta\!\left(\frac{\sqrt n}{\log n}\right).
\label{eq:intro-transmission-limit}
\end{equation}
A matching architecture-level converse proves that this order is not an artifact of the cumulative-sum construction or of a particular parameter choice. It is the transmission limit of the symmetric coordinatewise block-activity architecture formalized in Assumption~\ref{ass:symmetric-architecture}.

For context, we also analyze a deliberately relaxed full-horizon problem in which Bob waits for all $n$ observations and performs joint classical decoding after the same fixed Gaussian measurement. Within the same multiplicative circular-Gaussian modulation family, this reference supports $\Theta(\sqrt n)$ payload. It is not an unrestricted covert-capacity result and is not a componentwise ablation of the deadline constraint alone. Its role is narrower but useful: it shows that covertness and the selected fixed measurement, by themselves, do not force the additional $\log n$ factor. The comparison therefore isolates the logarithmic loss to the additional architectural requirements imposed by blockwise activity recovery, without claiming that local deadlines are the only possible source of that loss.

The local recovery component is related to quickest change detection. Page's cumulative-sum test and the minimax criteria of Lorden and Pollak describe false-alarm and delay tradeoffs for sequential decisions \cite{Page1954,lorden1971procedures,pollak1985optimal}, with later work treating finite windows and computational variants \cite{VeeravalliBanerjee2013,WangXie2023SequentialReview,PoorBook2009,Siegmund1985,Tartakovsky2014,Xie2023WindowLimitedCUSUM}. Burst synchronization and frame detection also include uncertain start times but generally do not impose a quantum covertness constraint \cite{HosseiniPerrins2014,Nasir2016,Zebarjadi2019,ProakisSalehi2008}. Timing uncertainty in covert communication has mostly been studied from Willie's side \cite{BashIgnorance2014,Huang2020LPD,Ramtin2025CovertQCDContinuous}, while covert quickest detection, quantum change-point detection, and covert target detection address different tasks and performance criteria \cite{LoBloch2026,FanizzaPRL2023,Tham2024QuantumLimitsCovertTarget}. Here the detector is at Bob, the block-activity vector identifies the message, and Willie is constrained through quantum relative entropy. The novelty is therefore neither cumulative-sum detection nor positional signaling alone, but the conversion of a block-local detection limit into an end-to-end covert information-transmission limit.

The paper makes four theorem-level contributions. First, Theorem~\ref{thm:exact-uniform} gives the exact finite-block Willie--Bob energy--information frontier and its unique uniform-energy optimizer; Theorem~\ref{thm:block-converse} and Corollary~\ref{cor:latency-theta} turn that frontier into the matching one-block latency limit, with Theorem~\ref{thm:block-achievability} providing a constructive block-reset detector. Second, Theorems~\ref{thm:finite-code}--\ref{thm:architecture-converse} lift the local law to a message code with worst-onset maximal-message covertness and vanishing maximal decoding error, yielding the matching transmission law in \eqref{eq:intro-transmission-limit} within the stated architecture. Third, Proposition~\ref{prop:global-reference} establishes the $\Theta(\sqrt n)$ relaxed full-horizon reference under the same modulation family and fixed measurement, clarifying what the block-activity law does and does not attribute to the local-recovery architecture. Fourth, Theorem~\ref{thm:finite-key} replaces ideal continuous shared randomness by public quadrature phase-shift keying codebooks selected through a finite secret key of total length $O(\sqrt n)$, while preserving the $\Theta(\sqrt n/\log n)$ payload order.

We state each result at its proven scope. Willie observes the lost-light weak-complementary output, Bob's main results begin after a fixed Gaussian measurement, and the matching end-to-end converse is architecture-specific rather than an unrestricted covert-capacity converse. In particular, the analysis concerns weight-constrained block activity; the constructive regime $N=\Theta(K)$ does not require a vanishing active fraction. Section~\ref{sec:conclusion} summarizes these limitations and the corresponding open problems.

The rest of the paper is organized as follows. Section~\ref{sec:model} defines the channel, receiver, activity-pattern code, and performance criteria. Sections~\ref{sec:exact-geometry} and~\ref{sec:cusum} establish the one-block frontier, converse, and achievability. Section~\ref{sec:coding} derives the end-to-end block-activity transmission limits and the full-horizon reference. Section~\ref{sec:finite-key} develops the finite-key realization, Section~\ref{sec:numerics} reports numerical checks, and Section~\ref{sec:conclusion} concludes the paper with scope and open questions. Longer proofs are collected in the appendices.

\section{Operational Model}
\label{sec:model}

We now fix the channel, receiver, code, and error criteria used throughout the paper. The three subsections describe the physical channel and shared randomness, Bob's measured channel, and the block-activity code.

\subsection{Thermal-loss channel and secret common randomness}

We begin with the optical channel and the shared modulation used in the exact finite-block analysis. Alice communicates over a single-mode thermal-loss channel for $n$ uses. Its power transmissivity is $\eta\in(0,1)$, and the environment is thermal with mean photon number $\bar n_B\ge0$ \cite{WeedbrookRMP2012,Gagatsos_2020}. Willie collects the unused output port of the beam splitter, also called the weak-complementary output of the thermal attenuator. For a mixed thermal environment, this output is not the full Stinespring complementary channel because it does not include a purification of the environment \cite{CarusoGiovannettiHolevo2006,Rosati2018}. All logarithms are natural, payloads are measured in nats, and each asymptotic statement specifies its limit.

We use the following conventions throughout. For density operators $\rho$ and
$\sigma$, the quantum relative entropy (QRE) is
\begin{equation}
D(\rho\|\sigma):=\Tr\!\left[\rho(\log\rho-\log\sigma)\right],
\label{eq:qre-definition}
\end{equation}
with value $+\infty$ when $\supp(\rho)\nsubseteq\supp(\sigma)$; the same
symbol $D(P\|Q)$ denotes classical Kullback--Leibler (KL) divergence. The
single-mode thermal state with mean photon number $N\ge0$ is
\begin{equation}
\tau_N:=\sum_{k=0}^{\infty}\frac{N^k}{(N+1)^{k+1}}\ket{k}\bra{k},
\label{eq:thermal-state-definition}
\end{equation}
where $\ket{k}$ is the $k$-photon Fock state. The annihilation and creation
operators are $a$ and $a^\dagger$, and the displacement operator is
\begin{equation}
\Disp(\alpha):=\exp(\alpha a^\dagger-\alpha^*a).
\label{eq:displacement-definition}
\end{equation}
We write $I_d$ for the $d\times d$ identity, $\1$ for the identity on the Fock space, $\Tr$ for trace, and
$\|A\|_2=(\Tr A^\dagger A)^{1/2}$ for the Hilbert--Schmidt norm. Alice's innocent input is $\tau_{\bar n_{\rm in}}$ with mean photon number $\bar n_{\rm in}\ge0$.

Alice and Bob share a sequence $S^n=(S_1,\ldots,S_n)$ of independent circularly symmetric complex Gaussian (CSCG) symbols
\begin{equation}
S_t\sim\CN(0,1),\qquad \E S_t=0,\quad \E|S_t|^2=1,
\label{eq:key-law}
\end{equation}
which is unavailable to Willie. At an active use with design energy $E_t\ge0$, Alice sends
\begin{equation}
\rho_{A,t|S_t=s}
 =\Disp(\sqrt{E_t}s)\tau_{\bar n_{\rm in}}\Disp^\dagger(\sqrt{E_t}s).
\label{eq:active-input}
\end{equation}
At an idle use she sends $\tau_{\bar n_{\rm in}}$. Thus, $E_t$ is the average incremental photon number, because $\E|\sqrt{E_t}S_t|^2=E_t$.

\begin{remark}[Common-randomness assistance]
The sequence in \eqref{eq:key-law} is an explicit model resource for the one-block and block-activity results in Sections~\ref{sec:exact-geometry}--\ref{sec:coding}, as well as the relaxed reference in Section~\ref{sec:global-benchmark}. These results are assisted by shared randomness and do not assume a finite key. Section~\ref{sec:finite-key} gives a separate QPSK construction with a finite sublinear key. Unless we condition on $S^n$, Bob's error probabilities average over both the shared sequence and the measurement noise. Willie's density operators are averaged over the shared sequence before their quantum relative entropy is evaluated.
\end{remark}

Let
\begin{equation}
\bar n_W=(1-\eta)\bar n_{\rm in}+\eta\bar n_B,
\qquad a_W=1-\eta.
\label{eq:willie-baseline}
\end{equation}
We assume $\bar n_W>0$ throughout; this excludes the singular vacuum-background case in which the relative entropy to the innocent state can be infinite. Conditioned on $S_t=s$, Willie's state is
\begin{equation}
\rho_{W,t|S_t=s}=\Disp(\sqrt{a_WE_t}s)\tau_{\bar n_W}\Disp^\dagger(\sqrt{a_WE_t}s).
\label{eq:willie-conditioned-state}
\end{equation}
Since a Gaussian mixture of displaced thermal states is thermal,
\begin{equation}
\bar\rho_{W}(E_t)
 :=\E_{S_t}[\rho_{W,t|S_t}]
 =\tau_{\bar n_W+a_WE_t}.
\label{eq:willie-exact-state}
\end{equation}

\subsection{Fixed Gaussian receiver at Bob}

With the channel fixed, we now write the classical observation law produced by Bob's Gaussian receiver. This law gives the likelihood ratio and the linear information slope used later.

Bob applies a fixed physically realizable Gaussian receiver \cite{Oh_2019,Holevo2021on}. We work with the classical model induced by this receiver. Let $C$ be the linear readout matrix and let $\Sigma_0\succ0$ be the idle output covariance, including measurement-added noise. Define
\begin{equation}
G=C^\top\Sigma_0^{-1}C\succeq0.
\label{eq:G-def}
\end{equation}
Here $\NN(\mu,\Sigma)$ denotes a real Gaussian law with mean $\mu$ and covariance $\Sigma$. Writing
\begin{equation}
X_t:=\operatorname{vec}(S_t)
=(\sqrt2\Re S_t,\sqrt2\Im S_t)^\top\sim\NN(0,I_2),
\end{equation}
Bob's conditional observation law is
\begin{align}
H_0:&\quad Y_t|S_t\sim\NN(0,\Sigma_0),\nonumber\\
H_1(E_t):&\quad Y_t|S_t\sim
\NN(\sqrt{\eta E_t}\,CX_t,\Sigma_0).
\label{eq:bob-laws}
\end{align}
Let $P_0$ and $P_{1,E_t,S_t}$ denote the corresponding idle and active conditional laws of $Y_t$, with densities $p_0$ and $p_{1,E_t,S_t}$. The conditional log-likelihood ratio (LLR) is
\begin{equation}
\ell_t(E_t)=
\log\frac{p_{1,E_t,S_t}(Y_t)}{p_0(Y_t)}.
\label{eq:llr}
\end{equation}
Define the random conditional slope and its mean by
\begin{align}
Z_t&:=\frac{\eta}{2}X_t^\top G X_t,\label{eq:Z-def}\\
\kappa&:=\E Z_t=\frac{\eta}{2}\Tr G.
\label{eq:kappa-def}
\end{align}
We assume $\kappa>0$, i.e., the selected receiver has nonzero sensitivity to Alice's displacement.
The chosen physical receiver fixes $(C,\Sigma_0)$; we do not optimize over arbitrary matrix pairs.

Conditioned on $S_t$, the equal-covariance Gaussian Kullback--Leibler (KL) formula \cite{zhang2023properties} gives, for every $E_t\ge0$,
\begin{equation}
D(P_{1,E_t,S_t}\|P_0)=Z_tE_t,
\label{eq:conditional-bob-kl}
\end{equation}
where the equality holds almost surely with respect to $S_t$.

Let $P_{S_t,Y_t}^{(1,E_t)}$ denote the joint law obtained by drawing $S_t\sim\CN(0,1)$ and then drawing $Y_t$ according to $P_{1,E_t,S_t}$. Similarly, let $P_{S_t,Y_t}^{(0)}$ denote the joint law obtained by drawing the same $S_t$ and then drawing $Y_t$ according to $P_0$. Since the marginal law of $S_t$ is the same under both hypotheses, the chain rule for relative entropy gives 
\begin{align}
D(P_{S_t,Y_t}^{(1,E_t)}\|P_{S_t,Y_t}^{(0)})
&=\E\!\left[D(P_{1,E_t,S_t}\|P_0)\right]\nonumber\\
&=\E[Z_t]E_t
=\kappa E_t.
\label{eq:joint-bob-kl}
\end{align}

Under $H_1$ and conditioned on $S_t$,
\begin{equation}
\ell_t(E_t)|S_t\sim\NN(Z_tE_t,2Z_tE_t).
\label{eq:llr-normal}
\end{equation}

\begin{remark}[Optional key-controlled homodyne instantiation]
\label{rem:key-controlled-homodyne}
The main results are stated for the fixed Gaussian receiver above.
Appendix~\ref{app:key-controlled-homodyne} records an optional homodyne
instantiation in which Bob uses the current shared symbol to choose the
local-oscillator direction before measurement. Within the class of
single-quadrature homodyne policies, the maximizing direction is
$w^*(X)=V^{-1}X/\|V^{-1}X\|$, and the corresponding average information
slope is $\kappa_{\rm hom}^{\rm ad}=(\eta/2)\Tr(V^{-1})$. Thus conclusions
that depend only on the linear information identity
\eqref{eq:joint-bob-kl} can be specialized by replacing $\kappa$ with
$\kappa_{\rm hom}^{\rm ad}$. This is a class-specific receiver
instantiation, not an optimization over general Gaussian or quantum
measurements.
\end{remark}

\subsection{Block-activity code}

We now make the message-bearing structure explicit. Each message is identified with a weight-constrained binary activity pattern across physical blocks, so recovering the message means recovering which blocks are active. The Hamming-weight constraint is the operative notion throughout; no vanishing-density assumption on $K/N$ is imposed. The definitions below specify the code together with the local deadlines and the global performance criteria.

The horizon is divided into $N$ blocks. Every block has length
\begin{equation}
B=L_{\rm start}+L,
\label{eq:block-length}
\end{equation}
where $L$ is the number of consecutive active samples and $L_{\rm start}$ is an onset-uncertainty window. An active segment may begin at any offset $\nu\in\{0,\ldots,L_{\rm start}\}$ such that its $L$ samples lie in the block. The onset is a deterministic nuisance parameter (or is randomized independently of the secret symbols and subkeys) and does not carry message information. Throughout the paper, ``localization'' means identifying the active block coordinates; Bob is not required to estimate the within-block onset $\nu$.

For integers $N$ and $K\le N$, define the activity-pattern codebook
\begin{align}
\cC_{N,K}&:=\{x^N\in\{0,1\}^N:\wt(x^N)\le K\},\nonumber\\
M_{N,K}&:=|\cC_{N,K}|=\sum_{j=0}^K\binom Nj.
\label{eq:activity-codebook}
\end{align}
Here $\wt(x^N)$ is the Hamming weight, or number of active coordinates. Let $\cM:=\{1,\ldots,M_{N,K}\}$ be the message set, with a one-to-one map from messages to patterns $x^N$. In an active block, Alice uses the energy vector $E^L=(E_1,\ldots,E_L)$. Bob resets his statistic at the start of each block and outputs a binary decision $\widehat x_b$. If $\widehat x^N$ belongs to the codebook, it identifies the decoded message $\widehat M$; any block error counts as a message error. We write $\Prb_m$ for the law induced by message $m$. In one-block analyses, $\Prb_0$ is the idle law and $\Prb_\nu$ is the law with active onset $\nu$.

For a message $m$ and an admissible vector $\boldsymbol\nu$ of active-block onsets, let $\rho_{W|m,\boldsymbol\nu}$ be Willie's state averaged over the secret common randomness, and let $\rho_{W|0}$ be the all-idle state. For a prescribed total covertness budget $\delta>0$, we use the stronger maximal-message, worst-onset covertness criterion
\begin{equation}
\max_{m\in\cM}\sup_{\boldsymbol\nu}
D(\rho_{W|m,\boldsymbol\nu}\|\rho_{W|0})\le\delta.
\label{eq:max-covertness}
\end{equation}
This criterion has a direct detection interpretation. By quantum Pinsker's inequality and the Helstrom bound \cite{WildeQIT}, for every fixed message and admissible onset vector the optimal equal-prior binary test between $\rho_{W|m,\boldsymbol\nu}$ and $\rho_{W|0}$ has advantage over random guessing at most $\min\{1/2,\sqrt{\delta/8}\}$. Thus \eqref{eq:max-covertness} controls Willie's distinguishability uniformly over messages and onsets, rather than only after averaging over them.
The maximal decoding error is
\begin{equation}
P_{\rm e}^{\max}:=
\max_{m\in\cM}\sup_{\nu_1,\ldots,\nu_N}
\Prb_m(\widehat M\ne m).
\label{eq:max-error}
\end{equation}
Both suprema range over the admissible onsets of the active blocks; idle blocks have no onset parameter.

\section{One-Block Energy--Information Frontier and Latency Limits}
\label{sec:exact-geometry}

Recovering the message-bearing activity pattern requires reliable decisions on its individual block coordinates. This section isolates one active segment, computes Willie's exact cost, solves the energy-allocation problem, and converts the resulting information frontier into a detector-independent lower bound on the latency of one such decision.

\subsection{Exact Willie cost}

We first express Willie's averaged state as a scalar energy cost. This exact function replaces a purely quadratic low-energy approximation.

For $x,y>0$, the relative entropy between single-mode thermal states is obtained by reducing quantum relative entropy to the geometric photon-number distributions \cite{WildeQIT,Parthasarathy2021}:
\begin{equation}
D(\tau_x\|\tau_y)
=(x+1)\log\frac{y+1}{x+1}+x\log\frac{x}{y}.
\label{eq:thermal-qre}
\end{equation}
We therefore define the exact per-use Willie cost
\begin{equation}
g(E):=D(\tau_{\bar n_W+a_WE}\|\tau_{\bar n_W}).
\label{eq:g-def}
\end{equation}

\begin{lemma}[Exact convex Willie cost]
\label{lem:g-properties}
The function $g:[0,\infty)\to[0,\infty)$ is continuous, strictly increasing on $(0,\infty)$, strictly convex, and unbounded; hence it has a well-defined inverse $g^{-1}:[0,\infty)\to[0,\infty)$. Let $x(E)=\bar n_W+a_WE$. Then
\begin{align}
g'(E)
&=a_W\log\frac{x(E)(\bar n_W+1)}{\bar n_W(x(E)+1)},
\label{eq:g-prime}\\
g''(E)
&=\frac{a_W^2}{x(E)(x(E)+1)}>0.
\label{eq:g-second}
\end{align}
Moreover,
\begin{equation}
g(E)=c_WE^2+O(E^3),
\qquad
c_W=\frac{a_W^2}{2\bar n_W(\bar n_W+1)},
\label{eq:g-expansion}
\end{equation}
and hence
\begin{equation}
g^{-1}(u)=\sqrt{\frac{u}{c_W}}+O(u),\qquad u\downarrow0.
\label{eq:g-inverse-expansion}
\end{equation}
\end{lemma}

\begin{proof}
Equation \eqref{eq:thermal-qre} follows because thermal states commute and have geometric photon-number distributions. Differentiating \eqref{eq:g-def} gives \eqref{eq:g-prime}--\eqref{eq:g-second}. Since $g'(0)=0$ and $g''(E)>0$, $g$ is strictly increasing for $E>0$ and strictly convex. Moreover, $g'(E)$ converges to the positive constant $a_W\log((\bar n_W+1)/\bar n_W)$, so $g(E)\to\infty$. Taylor expansion at zero yields \eqref{eq:g-expansion}; inversion of the local quadratic expansion gives \eqref{eq:g-inverse-expansion}.
\end{proof}

For an energy vector $E^L$, define $\rho_W(E^L):=\bigotimes_{t=1}^L\bar\rho_W(E_t)$ and $\rho_W(0):=\tau_{\bar n_W}^{\otimes L}$. For independent uses, QRE additivity \cite{WildeQIT} and \eqref{eq:willie-exact-state} give the exact identity
\begin{equation}
D(\rho_W(E^L)\|\rho_W(0))=\sum_{t=1}^L g(E_t).
\label{eq:exact-block-qre}
\end{equation}

\subsection{Exact optimal allocation}

With the exact cost in hand, we ask how Alice should spread a fixed budget across the $L$ active samples. The optimum is uniform at every finite $L$.

Let $P_{S^L,Y^L}^{(1,E^L)}$ and $P_{S^L,Y^L}^{(0)}$ denote the joint product laws induced by \eqref{eq:bob-laws} on an active segment with energy vector $E^L$ and on an idle segment, respectively.

\begin{theorem}[Exact finite-block energy--information frontier]
\label{thm:exact-uniform}
For every $L\ge1$ and $\delta_b>0$,
\begin{equation}
\sup_{\substack{E_t\ge0\\\sum_{t=1}^L g(E_t)\le\delta_b}}
D(P_{S^L,Y^L}^{(1,E^L)}\|P_{S^L,Y^L}^{(0)})
=\kappa Lg^{-1}\!\left(\frac{\delta_b}{L}\right).
\label{eq:exact-bob-frontier}
\end{equation}
The unique maximizing allocation is
\begin{equation}
E_1=\cdots=E_L=E_L^*:=g^{-1}\!\left(\frac{\delta_b}{L}\right).
\label{eq:exact-energy}
\end{equation}
Consequently,
\begin{equation}
\kappa L E_L^*
=\kappa\sqrt{\frac{\delta_bL}{c_W}}+O(1)
\qquad (L\to\infty).
\label{eq:bob-frontier-asymptotic}
\end{equation}
\end{theorem}

\begin{proof}
By \eqref{eq:joint-bob-kl}, the objective is $\kappa\sum_tE_t$. Let $\bar E=L^{-1}\sum_tE_t$. Strict convexity and Jensen's inequality yield
\begin{equation}
g(\bar E)\le\frac1L\sum_{t=1}^Lg(E_t)\le\frac{\delta_b}{L}.
\end{equation}
Since $g$ is increasing, $\bar E\le g^{-1}(\delta_b/L)$, and hence
\begin{equation}
\sum_{t=1}^LE_t\le Lg^{-1}(\delta_b/L).
\end{equation}
Equality in Jensen's inequality holds only for the uniform allocation, proving \eqref{eq:exact-bob-frontier}--\eqref{eq:exact-energy}. Equation \eqref{eq:bob-frontier-asymptotic} follows from \eqref{eq:g-inverse-expansion}.
\end{proof}

\subsection{A detector-independent latency converse}

The allocation theorem gives the largest information available to the fixed receiver in one block. We use this frontier to lower-bound the active-segment length for any detector that decides by the block deadline.

For $p,q\in(0,1)$, let
\begin{equation}
\db(p\|q)=p\log\frac pq+(1-p)\log\frac{1-p}{1-q}
\label{eq:binary-kl}
\end{equation}
be the binary relative entropy, and let
\begin{equation}
h_2(u)=-u\log u-(1-u)\log(1-u),\qquad u\in[0,1],
\label{eq:binary-entropy}
\end{equation}
with the convention $0\log0=0$.

\begin{theorem}[Receiver-restricted one-block converse]
\label{thm:block-converse}
Fix an onset $\nu\in\{0,\ldots,B-L\}$ in a block of length $B$, where the active segment has length $L$. Consider any energy allocation $E^L$ satisfying 
\begin{equation}
\sum_{t=1}^Lg(E_t)\le\delta_b,
\label{eq:block-budget}
\end{equation}
and any possibly sequential binary rule $\widehat X\in\{0,1\}$ whose final output is available by the end of the block. Let
\begin{equation}
\alpha=\Prb_0(\widehat X=1),
\qquad
\beta_\nu=\Prb_\nu(\widehat X=0).
\label{eq:block-errors}
\end{equation}
Then necessarily
\begin{equation}
\kappa Lg^{-1}\!\left(\frac{\delta_b}{L}\right)
\ge \db(1-\beta_\nu\|\alpha).
\label{eq:exact-latency-converse}
\end{equation}
Here $\db$ is understood in its standard extended-value sense at the boundary of $[0,1]^2$. In particular, if $\alpha\le\epsilon_b$ and $\sup_\nu\beta_\nu\le\epsilon_b<1/2$, then
\begin{equation}
\kappa Lg^{-1}\!\left(\frac{\delta_b}{L}\right)
\ge(1-2\epsilon_b)\log\frac{1-\epsilon_b}{\epsilon_b}.
\label{eq:symmetric-error-converse}
\end{equation}
For fixed $\delta_b>0$, as $\epsilon_b\downarrow0$ this implies
\begin{equation}
L\ge
\frac{c_W}{\kappa^2\delta_b}
\left(\log\frac1{\epsilon_b}\right)^2(1-o(1)).
\label{eq:latency-converse-asymptotic}
\end{equation}
\end{theorem}

\begin{proof}
Let $Q_\nu$ and $Q_0$ be the joint laws of all common-randomness symbols and receiver observations available in the block under onset $\nu$ and under the idle hypothesis, respectively. The pre-onset and post-segment observations have identical conditional laws under the two hypotheses and are independent of the active-segment observations. Therefore, by independence and \eqref{eq:joint-bob-kl},
\begin{equation}
D(Q_\nu\|Q_0)=\kappa\sum_{t=1}^LE_t.
\end{equation}
The final output of any sequential rule is a measurable function of the block data, so binary data processing gives
\begin{equation}
\kappa\sum_{t=1}^LE_t
\ge \db(1-\beta_\nu\|\alpha).
\end{equation}
Theorem \ref{thm:exact-uniform} upper bounds the left-hand side by $\kappa Lg^{-1}(\delta_b/L)$, proving \eqref{eq:exact-latency-converse}. The binary divergence is minimized over $1-\beta_\nu\ge1-\epsilon_b$ and $\alpha\le\epsilon_b$ at $(1-\epsilon_b,\epsilon_b)$, yielding \eqref{eq:symmetric-error-converse}. The right-hand side diverges as $\epsilon_b\downarrow0$, so every feasible sequence has $L\to\infty$ and hence $\delta_b/L\to0$. Substitution of \eqref{eq:g-inverse-expansion} and solution for $L$ gives \eqref{eq:latency-converse-asymptotic}.
\end{proof}

\section{CUSUM Achievability: False Alarms and Missed Detections}
\label{sec:cusum}

The converse above gives the latency required to recover one activity coordinate under a block deadline. We now analyze a simple block-reset CUSUM and show that it attains the same leading order. We bound false alarms and missed detections separately, then combine the bounds into one design rule.

Bob resets the reflected CUSUM statistic \cite{Page1954,PoorBook2009,Siegmund1985,Tartakovsky2014} at the beginning of each block:
\begin{equation}
W_t=\max\{0,W_{t-1}+\ell_t(E_L^*)\},
\qquad W_0=0,
\label{eq:cusum-recursion}
\end{equation}
with stopping time
\begin{equation}
T_h=\inf\{t\ge1:W_t\ge h\}.
\label{eq:cusum-stopping}
\end{equation}
The block is declared active iff $T_h\le B$.

\subsection{Finite-horizon false alarms}

For an idle block, we first bound the probability that CUSUM crosses the threshold. The bound holds for each shared sequence and therefore also after averaging.

\begin{lemma}[Block false-alarm bound]
\label{lem:false-alarm}
For every threshold $h>0$, every blocklength $B$, and every realization $S^B$ of the secret sequence,
\begin{equation}
\Prb_0(T_h\le B\mid S^B)\le Be^{-h}.
\label{eq:false-alarm-bound}
\end{equation}
The same bound holds after averaging over $S^B$.
\end{lemma}

\begin{proof}
For every possible segment start $k\in\{1,\ldots,B\}$, define $M_{k,k-1}=1$ and
\begin{equation}
M_{k,t}=\exp\left(\sum_{i=k}^t\ell_i(E_L^*)\right),
\qquad t\ge k.
\end{equation}
Conditioned on $S^B$, $(M_{k,t})_{t\ge k-1}$ is a nonnegative martingale under $H_0$ with initial mean one because every exponentiated LLR increment satisfies $\E_0[e^{\ell_i(E_L^*)}\mid S_i]=1$. The reflected CUSUM representation implies that $T_h\le B$ only if, for some $k\le t\le B$, $M_{k,t}\ge e^h$. Ville's inequality and a union bound over $k$ give
\begin{align}
\Prb_0(T_h\le B\mid S^B)
&\le\sum_{k=1}^B\Prb_0\left(
\max_{k\le t\le B}M_{k,t}\ge e^h\mid S^B\right)\nonumber\\
&\le Be^{-h}.
\end{align}
\end{proof}

\subsection{Missed detections}

For an active block, we next bound the miss probability. We separate an unusually small key-induced drift from the Gaussian noise under a typical drift.

For $\xi\in(0,1)$, recall from \eqref{eq:G-def} and \eqref{eq:kappa-def}, define the lower-tail rate of the key-induced slope by
\begin{equation}
I_-(\xi):=
\sup_{\theta\ge0}
\left\{
\frac12\log\det(I_2+\theta\eta G)
-\theta(1-\xi)\kappa
\right\}>0.
\label{eq:key-rate-function}
\end{equation}
Indeed, the Gaussian quadratic-form identity gives $\E[e^{-\theta Z_t}]=\det(I_2+\theta\eta G)^{-1/2}$ for $\theta\ge0$. The positivity follows because the derivative of the objective at $\theta=0$ is $\xi\kappa>0$. Markov's inequality applied to $e^{-\theta\sum_t Z_t}$ and optimization over $\theta$ give
\begin{equation}
\Prb\left(\frac1L\sum_{t=1}^LZ_t<(1-\xi)\kappa\right)
\le e^{-LI_-(\xi)}.
\label{eq:key-lower-tail}
\end{equation}
For an isotropic receiver with $\eta G=q_{\rm iso}I_2$, $Z_t$ is exponential with mean $\kappa=q_{\rm iso}$ and
\begin{equation}
I_-(\xi)=-\xi-\log(1-\xi).
\label{eq:isotropic-rate}
\end{equation}

\begin{lemma}[Uniform-onset missed-detection bound]
\label{lem:miss}
Fix $\xi\in(0,1)$ and $\gamma>0$. Use the exact uniform energy $E_L^*$ in \eqref{eq:exact-energy}. If
\begin{equation}
(1-\xi)\kappa L E_L^*\ge(1+\gamma)h,
\label{eq:drift-condition}
\end{equation}
then, for every active onset allowed by \eqref{eq:block-length},
\begin{equation}
\Prb_\nu(T_h>\nu+L)
\le e^{-LI_-(\xi)}
+\exp\left[-\frac{\gamma^2}{4(1+\gamma)}h\right].
\label{eq:miss-bound}
\end{equation}
Here $T_h$ is measured from the beginning of the block.
\end{lemma}

\begin{proof}[Proof sketch]
Condition on the event that the average key-induced slope is at least $(1-\xi)\kappa$. The active-segment log-likelihood sum is then Gaussian. Its mean is at least $(1+\gamma)h$, and its variance is twice its mean. A Gaussian tail bound gives the second term in \eqref{eq:miss-bound}. Equation \eqref{eq:key-lower-tail} gives the first term. The reflected CUSUM is at least the ordinary sum that starts at the true onset, so crossing of that ordinary sum implies detection by $\nu+L$. Appendix~\ref{app:proof-miss} gives the full calculation.
\end{proof}

\begin{theorem}[Finite-block one-block achievability]
\label{thm:block-achievability}
Fix target errors $\alpha_b,\beta_b\in(0,1)$, a covertness budget $\delta_b$, and constants $\xi\in(0,1)$ and $\gamma>0$. Let
\begin{equation}
c_\gamma=\frac{\gamma^2}{4(1+\gamma)}.
\end{equation}
If $L$ and $h$ satisfy
\begin{align}
h&\ge\log\frac{B}{\alpha_b},
\label{eq:h-fa}\\
h&\ge c_\gamma^{-1}\log\frac{2}{\beta_b},
\label{eq:h-md}\\
L&\ge I_-(\xi)^{-1}\log\frac{2}{\beta_b},
\label{eq:L-key}\\
(1-\xi)\kappa Lg^{-1}\!\left(\frac{\delta_b}{L}\right)
&\ge(1+\gamma)h,
\label{eq:implicit-latency-ach}
\end{align}
then the block-reset CUSUM has false-alarm probability at most $\alpha_b$, worst-onset missed-detection probability at most $\beta_b$, and exact Willie QRE equal to $\delta_b$ over the active segment.
\end{theorem}

\begin{proof}
Use $E_L^*=g^{-1}(\delta_b/L)$. Exact covertness follows from $Lg(E_L^*)=\delta_b$. Equations \eqref{eq:h-fa} and Lemma \ref{lem:false-alarm} control false alarm. Equations \eqref{eq:h-md}--\eqref{eq:implicit-latency-ach} and Lemma \ref{lem:miss} control missed detection.
\end{proof}

Let $L^*(\epsilon_b,\delta_b)$ denote the minimum active-segment length for which there exist an allocation obeying \eqref{eq:block-budget} and a block-end binary rule with false-alarm probability and worst-onset miss probability both at most $\epsilon_b$.

\begin{corollary}[Matching-order one-block latency limit]
\label{cor:latency-theta}
Assume that $L_{\rm start}\le\lambda L$ for a fixed
$\lambda\ge0$. Let the false-alarm and worst-onset missed-detection
targets both equal $\epsilon_b$. For every fixed $\delta_b>0$,
\begin{equation}
L^*(\epsilon_b,\delta_b)
=
\Theta\!\left(
\frac{c_W}{\kappa^2\delta_b}
\log^2\frac1{\epsilon_b}
\right),
\qquad
\epsilon_b\downarrow0.
\label{eq:theta-latency}
\end{equation}
\end{corollary}

\begin{proof}
The lower bound follows directly from
\eqref{eq:latency-converse-asymptotic}. Indeed, that bound applies to every detector whose false-alarm probability and worst-onset missed-detection probability are at most $\epsilon_b$. 

We next construct a feasible block-reset CUSUM. Fix
$\xi\in(0,1)$ and $\gamma>0$, and let 
\[
c_\gamma:=\frac{\gamma^2}{4(1+\gamma)}.
\]
Choose a constant $H>\max\{1,c_\gamma^{-1}\}$
and then choose a dimensionless constant $C$ such that 
\begin{equation}
(1-\xi)\sqrt C>(1+\gamma)H.
\label{eq:latency-C-choice}
\end{equation}
Set
\begin{equation}
L=
\left\lceil
C\frac{c_W}{\kappa^2\delta_b}(\log\frac1{\epsilon_b})^2
\right\rceil,
\qquad
h=H\log\frac1{\epsilon_b}.
\label{eq:latency-cusum-choice}
\end{equation}

We verify the four conditions of
Theorem~\ref{thm:block-achievability}. First,
$B=L_{\rm start}+L\le(1+\lambda)L$, and hence 
\[
\log\frac{B}{\epsilon_b}
=
\log\frac1{\epsilon_b}+O\big(\log( \log\frac1{\epsilon_b})\big).
\]
Since $H>1$, we have $h\ge\log\frac{B}{\epsilon_b}$ for all sufficiently small $\epsilon_b$. Thus the false-alarm condition
\eqref{eq:h-fa} holds. 
Second, since $H>c_\gamma^{-1}$, we have $h\ge
c_\gamma^{-1}\log\frac{2}{\epsilon_b}$ for sufficiently small $\epsilon_b$, so \eqref{eq:h-md} holds. 
Third, $L$ grows as $R^2$, whereas the right-hand side of \eqref{eq:L-key} grows only as $\log\frac1{\epsilon_b}$. Therefore, 
\[
L\ge I_-(\xi)^{-1}\log\frac{2}{\epsilon_b}
\]
for sufficiently small $\epsilon_b$. 

Finally, using \eqref{eq:g-inverse-expansion}, \eqref{eq:latency-C-choice}, and \eqref{eq:latency-cusum-choice}, 
\begin{align}
(1-\xi)\kappa Lg^{-1}\!\left(\frac{\delta_b}{L}\right)
&=
(1-\xi)\kappa
\left(
\sqrt{\frac{\delta_bL}{c_W}}+O(1)
\right)\nonumber\\
&=
(1-\xi)\sqrt C\,\log\frac1{\epsilon_b}+O(1)\\
&>(1+\gamma)h+O(1). 
\end{align}
for all sufficiently small $\epsilon_b$. 

Theorem~\ref{thm:block-achievability} now guarantees that this block-reset CUSUM has false-alarm probability at most $\epsilon_b$, worst-onset missed-detection probability at most $\epsilon_b$, and Willie relative-entropy cost $\delta_b$. It is therefore a feasible rule in the definition of $L^*(\epsilon_b,\delta_b)$, so
\[
L^*(\epsilon_b,\delta_b)\le L.
\]
Combining the lower and upper bounds proves \eqref{eq:theta-latency}.
\end{proof}

\section{End-to-End Block-Activity Transmission Limits}
\label{sec:coding}

The preceding sections characterize the cost of recovering one coordinate of the weight-constrained activity pattern. We now ask how many such patterns can serve as reliably distinguishable covert messages over a growing horizon. Because all active coordinates share one total covertness budget and every block decision contributes to message recovery, neither the per-block budget nor the per-block error target can remain fixed as the number of message-bearing coordinates grows.

This section converts the one-block latency law into the end-to-end payload law of the activity-pattern code. We first lift the local guarantees to a finite-horizon message code, then scale $N$, $K$, $L$, and the detection threshold jointly, and finally prove that the resulting payload order is optimal within the stated blockwise architecture.

\subsection{From Block Decisions to Message Reliability}

At a fixed horizon, a message is recovered only when its complete weight-constrained activity pattern is recovered. The total Willie budget must therefore be shared across the active blocks, while each local decision must be reliable enough to control the error of the whole pattern. The next theorem combines these requirements and gives per-message covertness together with an end-to-end decoding-error bound.

\begin{theorem}[Common-randomness-assisted block-activity code]
\label{thm:finite-code}
Fix integers $N\ge1$, $1\le K\le N$, $L\ge1$, and $L_{\rm start}\ge0$, a total covertness budget $\delta>0$, and constants $\xi\in(0,1)$ and $\gamma>0$. Use the codebook $\cC_{N,K}$ in \eqref{eq:activity-codebook}, assign each active block budget
\begin{equation}
\delta_b=\frac{\delta}{K},
\label{eq:equal-budget}
\end{equation}
and use energy
\begin{equation}
E_L^*=g^{-1}\!\left(\frac{\delta}{KL}\right)
\label{eq:code-energy}
\end{equation}
in each active sample. Let $\alpha_B$ and $\beta_B$ denote the one-block bounds
\begin{align}
\alpha_B&=Be^{-h},
\label{eq:alphaB}\\
\beta_B&=e^{-LI_-(\xi)}+e^{-c_\gamma h},
\label{eq:betaB}
\end{align}
such that the drift condition
\begin{equation}
(1-\xi)\kappa Lg^{-1}\!\left(\frac{\delta}{KL}\right)
\ge(1+\gamma)h
\label{eq:code-drift}
\end{equation}
holds. Then the code has
\begin{equation}
M=\sum_{j=0}^K\binom Nj
\label{eq:code-size}
\end{equation}
messages and satisfies
\begin{align}
\max_m\sup_{\boldsymbol\nu}D(\rho_{W|m,\boldsymbol\nu}\|\rho_{W|0})&\le\delta,
\label{eq:code-covertness}\\
P_{\rm e}^{\max}&\le N\alpha_B+K\beta_B.
\label{eq:code-error}
\end{align}
\end{theorem}

\begin{proof}
A message of Hamming weight $w\le K$ contains exactly $wL$ active uses. Additivity and \eqref{eq:code-energy} give
\begin{equation}
D(\rho_{W|m,\boldsymbol\nu}\|\rho_{W|0})
=wL g(E_L^*)=w\frac{\delta}{K}\le\delta
\end{equation}
for every admissible onset vector $\boldsymbol\nu$.
Thus covertness holds separately for every message; by convexity it also holds for the message-averaged Willie state.

For reliability, every idle block has false-alarm probability at most $\alpha_B$ by Lemma \ref{lem:false-alarm}, and every active block has worst-onset miss probability at most $\beta_B$ by Lemma \ref{lem:miss}. A union bound over the at most $N$ idle/active block decisions gives \eqref{eq:code-error}. Exact recovery of the activity vector identifies the unique message.
\end{proof}

Theorem \ref{thm:finite-code} is the bridge from one-block detection to communication: it shows how the local covertness and error guarantees combine when the $N$ block decisions jointly represent one message.

\subsection{Achievability Under Local Block Deadlines}

The finite-horizon theorem is valid for fixed $(N,K,L,h)$, but it does not yet say how many nats can be carried as the total horizon grows. We now scale these parameters jointly. The key effect is that increasing $K$ splits the fixed Willie budget across more active blocks, while increasing $N$ forces the individual block decisions to become more reliable. These two requirements determine the block length and, together with the number of weight-constrained activity patterns, yield the end-to-end payload law.

\begin{theorem}[Local-deadline block-activity achievability]
\label{thm:scaling-ach}
Fix $\delta>0$, $\varrho\ge2$, and $\lambda\ge0$. There exist constants $A,H>0$ such that the following sequence of codes is covert and has vanishing maximal error. For integers $K\to\infty$, set
\begin{align}
N&=\lceil\varrho K\rceil,\nonumber\\
L&=\left\lceil A K(\log K)^2\right\rceil,\nonumber\\
L_{\rm start}&\le\lambda L,\nonumber\\
h&=H\log K.
\label{eq:scaling-choice}
\end{align}
Then, with total channel uses $n=NB$,
\begin{equation}
\log M=\Theta\left(\frac{\sqrt n}{\log n}\right).
\label{eq:scaling-ach-result}
\end{equation}
All logarithms and payloads in this theorem are in nats.
\end{theorem}

\begin{proof}
Fix any $\xi\in(0,1)$ and $\gamma>0$. The exact covertness statement follows from Theorem \ref{thm:finite-code}. Since $N=\Theta(K)$ and $B=\Theta(K\log^2K)$,
\begin{equation}
NBe^{-h}=O(K^2\log^2K\,K^{-H}),
\end{equation}
which vanishes for sufficiently large $H$. Also,
\begin{equation}
Ke^{-LI_-(\xi)}\to0,
\end{equation}
and $Ke^{-c_\gamma h}\to0$ if $Hc_\gamma>1$.

It remains to enforce \eqref{eq:code-drift}. By \eqref{eq:g-inverse-expansion},
\begin{align}
L g^{-1}\!\left(\frac{\delta}{KL}\right)
&=\sqrt{\frac{\delta L}{c_WK}}(1+o(1))\nonumber\\
&=\sqrt{\frac{\delta A}{c_W}}\log K\,(1+o(1)).
\end{align}
Thus $A$ can be chosen so that \eqref{eq:code-drift} holds for the selected $H,\xi,\gamma$.

Finally,
\begin{equation}
n=NB=\Theta(K^2\log^2K),
\end{equation}
whereas Stirling's approximation gives
\begin{equation}
\log M\ge\log\binom{\lceil\varrho K\rceil}{K}
=\varrho K h_2(1/\varrho)+o(K)=\Theta(K).
\end{equation}
The reverse bound $\log M\le N\log2=O(K)$ is immediate. Since $K=\Theta(\sqrt n/\log n)$, the result follows.
\end{proof}

\subsection{Order Optimality Within the Symmetric Blockwise Architecture}

The preceding achievability result shows that the proposed
block-activity architecture supports 
\[
\log M=\Theta\!\left(\frac{\sqrt n}{\log n}\right)
\]
nats while recovering the message from separate active-or-idle decisions across the blocks.

We now prove the converse half of this scaling law. The code parameters, active-segment energy allocation, and block-local decision rule may vary with $n$, but the active symbols remain within the shared-CSCG displaced-thermal signaling family for which the exact one-block frontier was derived. Within the symmetric weight-constrained architecture formalized below, no reliable scheme can achieve a larger payload order. Hence the preceding construction is order-optimal at the level actually proved; the result does not apply to arbitrary covert codes, arbitrary signaling states, or full-horizon joint decoding.

\begin{assumption}[Symmetric coordinatewise architecture]
\label{ass:symmetric-architecture}
For every $n$, the code uses $N=N_n$ equal blocks and the full codebook $\cC_{N,K}$ with $1\le K=K_n\le N/2$. Each active block uses the displaced-thermal, shared-CSCG signaling family of Section~\ref{sec:model} with QRE budget $\delta/K$ and the same active-segment energy vector (which may depend on $n$ and may be nonuniform within the segment). All blocks employ the same fixed Gaussian receiver of Section~\ref{sec:model} and the same deterministic block-local binary decision rule, reset at each block boundary. Thus each block decision is a function only of that block's receiver observations and common-randomness symbols. Channel noise and common randomness are independent across blocks. The maximal decoding error $\epsilon_n$ tends to zero.
\end{assumption}

\begin{theorem}[Matching block-activity transmission converse]
\label{thm:architecture-converse}
Under Assumption \ref{ass:symmetric-architecture}, for fixed total covertness budget $\delta$,
\begin{equation}
\log M=O\left(\frac{\sqrt n}{\log n}\right).
\label{eq:architecture-upper}
\end{equation}
Together with Theorem \ref{thm:scaling-ach}, the payload of the symmetric full weight-constrained pattern architecture in Assumption~\ref{ass:symmetric-architecture} is therefore
\begin{equation}
\log M=\Theta\left(\frac{\sqrt n}{\log n}\right).
\label{eq:architecture-theta}
\end{equation}
\end{theorem}

\begin{proof}[Proof sketch]
The all-idle message forces the false-alarm probability of one block to be $O(\epsilon_n/N)$. A weight-$K$ message at the worst onset similarly forces the miss probability to be $O(\epsilon_n/K)$. Binary data processing then requires one-block divergence of order $\log N$. The exact one-block frontier gives $L=\Omega(K\log^2N)$. Combining this bound with $n\ge NL$ and the usual bound on $\sum_{j\le K}\binom Nj$ proves \eqref{eq:architecture-upper}. Appendix~\ref{app:proof-architecture-converse} gives the details and treats both ranges of $N$.
\end{proof}

\subsection{Relaxed Full-Horizon Reference and the Square-Root Law}
\label{sec:global-benchmark}

The preceding results characterize messages carried by block activity when Bob must recover the activity vector through block-local decisions and deadlines. To separate this architecture-level penalty from the baseline covert scaling associated with the same physical modulation and measurement, we next analyze a deliberately relaxed full-horizon reference. Bob observes the entire length-$n$ horizon before decoding, and neither block-local recovery nor a stopping-time or delay criterion is imposed. 

\begin{proposition}[Relaxed full-horizon square-root reference]
\label{prop:global-reference}
Consider the relaxed full-horizon model defined in Appendix~\ref{app:benchmark-proofs}, with the same shared Gaussian modulation, the same fixed Gaussian measurement at Bob, and a fixed per-message Willie relative-entropy budget $\delta>0$. Within that modulation family, the largest reliably decodable payload satisfies 
\begin{equation}
 \log M=\Theta(\sqrt n).
 \label{eq:global-reference-scaling}
\end{equation}
\end{proposition}

\begin{proof}
The converse is Theorem~\ref{thm:global-fixed-converse}, and the matching first-order construction is Theorem~\ref{thm:dense-achievability}; both are proved in Appendix~\ref{app:benchmark-proofs}.
\end{proof}

The two scaling laws answer different operational questions. The $\Theta(\sqrt n/\log n)$ result is order-optimal for the proposed architecture while retaining block-local activity recovery, family-wise reliability, and a latency guarantee. Proposition~\ref{prop:global-reference} allows full-horizon coding and decoding and therefore is not a componentwise ablation of the deadline constraint. Its narrower purpose is to show that, within the stated modulation family, covertness and the selected fixed measurement by themselves do not force the additional $\log n$ factor; Section~\ref{sec:conclusion} gives the corresponding interpretation and limitations.

\section{Finite-Key Block-Activity Realization}
\label{sec:finite-key}

Having established the payload law for the weight-constrained activity-pattern code, we now return to the shared-randomness assumption used to realize its active blocks. The one-block and block-activity results above use an ideal continuous circular Gaussian sequence shared by Alice and Bob. This section replaces it with finite secret indices while preserving the $\Theta(\sqrt n/\log n)$ payload order. Each physical block has a public quadrature phase-shift keying (QPSK) codebook, and a secret subkey selects one sequence. Because Willie no longer sees an exactly thermal average state, we compare the QPSK average with the thermal target and prove a soft-covering bound for a finite public codebook. General quantum soft-covering and channel-resolvability results give broader context \cite{ChengGao2024,HayashiChengGao2025}; the proof here is self-contained and uses the symmetry of displaced thermal QPSK states.

Let $N_W=\bar n_W>0$ and write $\tau:=\tau_{N_W}$. For any operator $A$ with finite weighted norm, define
\begin{equation}
\norm{A}_{2,\tau}:=\norm{\tau^{-1/4}A\tau^{-1/4}}_2,
\label{eq:weighted-hs-norm}
\end{equation}
and write $A(E)=O_{2,\tau}(E^j)$ when $\norm{A(E)}_{2,\tau}\le C E^j$ for all sufficiently small $E$. Define
\begin{align}
 \rho_r(E)&=\Disp(\sqrt{a_WE}\,i^r)\tau_{N_W}
 \Disp^\dagger(\sqrt{a_WE}\,i^r),\nonumber\\
 &\hspace{4em}r=0,1,2,3,\\
 \sigma_E&=\frac14\sum_{r=0}^3\rho_r(E),\qquad
 q(E)=D(\sigma_E\|\tau_{N_W}).
 \label{eq:qpsk-states}
\end{align}

\begin{lemma}[Local QPSK covertness cost]
\label{lem:qpsk-cost}
As $E\downarrow0$,
\begin{align}
 q(E)&=g(E)+O(E^4)=c_WE^2+O(E^3),\nonumber\\
 c_W&=\frac{a_W^2}{2N_W(N_W+1)}.
 \label{eq:qpsk-cost-expansion}
\end{align}
In particular, $q(E)>0$ for every sufficiently small $E>0$.
\end{lemma}

\begin{proof}[Proof sketch]
The QPSK average and $\tau_{N_W+a_WE}$ have the same mean photon number, which gives the projection identity \eqref{eq:qpsk-projection}. QPSK phase averaging matches the order-$E$ thermal perturbation and cancels odd displacement orders. Appendix~\ref{app:qpsk-regularity} shows that the remaining state difference is $O(E^2)$ in the weighted norm, so the projection remainder is $O(E^4)$. Appendix~\ref{app:proof-qpsk-cost} gives the full proof.
\end{proof}

Define the fixed-thermal-reference collision quantities
\begin{align}
 Q_{2,\tau}(E)&=\frac14\sum_{r=0}^3
 \Tr\!\left[
 \left(\tau_{N_W}^{-1/4}\rho_r(E)\tau_{N_W}^{-1/4}\right)^2
 \right],\nonumber\\
 R_{2,\tau}(E)&=
 \Tr\!\left[
 \left(\tau_{N_W}^{-1/4}\sigma_E\tau_{N_W}^{-1/4}\right)^2
 \right].
 \label{eq:qpsk-collision}
\end{align}

\begin{lemma}[Fixed-reference collision expansion]
\label{lem:collision-expansion}
For $N_W>0$,
\begin{align}
 Q_{2,\tau}(E)
 &=1+\frac{2a_W}{\sqrt{N_W(N_W+1)}}E+O(E^2),
 \label{eq:collision-expansion}\\
 R_{2,\tau}(E)&=1+O(E^2).
 \label{eq:average-collision-expansion}
\end{align}
Consequently, finite constants $d_2,E_0>0$ exist such that
\begin{equation}
 Q_{2,\tau}(E)\le e^{d_2E},\qquad 0\le E\le E_0.
 \label{eq:collision-exponential}
\end{equation}
Moreover,
\begin{align}
 Q_{2,\tau}(E)-R_{2,\tau}(E)
 &=\frac14\sum_{r=0}^3
 \norm{\tau_{N_W}^{-1/4}
 [\rho_r(E)-\sigma_E]\tau_{N_W}^{-1/4}}_2^2.
 \label{eq:qpsk-collision-variance}
\end{align}
\end{lemma}

\begin{proof}[Proof sketch]
Expand each displaced thermal state in the fixed thermal weighted Hilbert--Schmidt norm. The trace-zero property removes the linear cross term, and QPSK averaging cancels the remaining odd powers. A geometric sum in the Fock basis gives the coefficient in \eqref{eq:collision-expansion}. The centered-square identity is exact. Appendix~\ref{app:proof-collision-expansion} gives the details.
\end{proof}

\begin{lemma}[Relative-entropy collision soft covering]
\label{lem:collision-soft-covering}
Let a finite ensemble $\{p(x),\rho_x\}$ have average $\sigma$, and let $\tau$
be full-support (faithful) with $\sigma\succeq c\tau$ for some $c>0$. Define
\begin{align}
 Q_{2,\tau}&=\sum_xp(x)
 \Tr[(\tau^{-1/4}\rho_x\tau^{-1/4})^2],\nonumber\\
 R_{2,\tau}&=\Tr[(\tau^{-1/4}\sigma\tau^{-1/4})^2].
\end{align}
For $J$ independent samples and
$\widehat\sigma=J^{-1}\sum_{j=1}^J\rho_{X_j}$,
\begin{equation}
 \E_{\cC}D(\widehat\sigma\|\sigma)
 \le\frac{Q_{2,\tau}-R_{2,\tau}}{cJ}
 \le\frac{Q_{2,\tau}}{cJ}.
 \label{eq:collision-soft-covering}
\end{equation}
For an $L$-fold product ensemble with reference $T=\tau^{\otimes L}$, the
same bound holds with $c^L$, $Q_{2,\tau}^L$, and $R_{2,\tau}^L$.
\end{lemma}

\begin{proof}[Proof sketch]
We first bound relative entropy by the order-two sandwiched R\'enyi divergence. The condition $\sigma\succeq c\tau$ lets us use the fixed thermal weight. The centered second moment of the empirical mixture gives the factor $1/J$, and tensor products give the product form. Appendix~\ref{app:proof-soft-covering} gives the details.
\end{proof}

\begin{lemma}[QPSK group Pythagorean identity]
\label{lem:qpsk-pythagorean}
For $\bm s=(s_1,\ldots,s_L)\in\{0,1,2,3\}^L$, let $\rho_{\bm s}(E)=\bigotimes_{t=1}^L\rho_{s_t}(E)$, $\Sigma_E=\sigma_E^{\otimes L}$, and $T=\tau_{N_W}^{\otimes L}$. For every probability distribution $\pi$ on QPSK sequences and $\widehat\Sigma=\sum_{\bm s}\pi(\bm s)\rho_{\bm s}(E)$,
\begin{equation}
 D(\widehat\Sigma\|T)
 =D(\widehat\Sigma\|\Sigma_E)+D(\Sigma_E\|T).
 \label{eq:qpsk-pythagorean}
\end{equation}
\end{lemma}

\begin{proof}
The local phase group $\mathbb Z_4^L$ acts transitively on the product states. Both $\Sigma_E$ and $T$ are invariant, so $A=\log\Sigma_E-\log T$ is invariant. Hence $\Tr[\rho_{\bm s}(E)A]$ is independent of $\bm s$ and equals its uniform average $\Tr[\Sigma_EA]=D(\Sigma_E\|T)$. All three relative entropies are finite by the thermal domination established below. Expanding the two relative entropies proves \eqref{eq:qpsk-pythagorean}.
\end{proof}

\begin{theorem}[Finite-key block-activity transmission achievability]
\label{thm:finite-key}
Fix $\delta>0$, $\varrho\ge2$, $\lambda\ge0$, $\xi\in(0,1)$, and $\gamma>0$.
There exist finite constants $A,H,r_{\rm key}>0$ such that, for integers
$K\to\infty$, the choices
\begin{align}
 N_K&=\lceil\varrho K\rceil,&
 L_K&=\left\lceil A K(\log K)^2\right\rceil,\nonumber\\
 0\le L_{{\rm start},K}&\le\lambda L_K,&
 B_K&=L_K+L_{{\rm start},K},\nonumber\\
 h_K&=H\log K
 \label{eq:finite-key-scaling}
\end{align}
admit deterministic public QPSK codebooks of size
\begin{equation}
 J_K=2^{\lceil r_{\rm key}\log_2K\rceil}
 \label{eq:finite-key-codebook-size}
\end{equation}
for each physical block and independent uniform secret codebook indices such that
\begin{align}
 \max_m\sup_{\bm\nu}D(\rho_{W|m,\bm\nu}\|\rho_{W|0})&\le\delta,
 \label{eq:finite-key-covertness}\\
 P_{\rm e}^{\max}&\longrightarrow0.
 \label{eq:finite-key-reliability}
\end{align}
The total secret-key length and payload satisfy
\begin{align}
 \ell_K&=N_K\log_2J_K=O(K\log K),\nonumber\\
 \ell_K&=O(\sqrt{n_K}),\qquad \ell_K/n_K\to0,
 \label{eq:finite-key-length}\\
 \log M_K&=\Theta(K)
 =\Theta\!\left(\frac{\sqrt{n_K}}{\log n_K}\right),
 \label{eq:finite-key-payload}\\
 n_K&=N_KB_K.\nonumber
\end{align}
\end{theorem}

\begin{proof}[Proof sketch]
Use the energy with a small safety margin in \eqref{eq:finite-key-energy}. Lemma~\ref{lem:collision-soft-covering}, Markov's inequality, and a union bound give one set of public QPSK codebooks that works for every block and onset at Willie. Hoeffding's inequality over random public-codebook generation and a second union bound select the same deterministic codebooks with sufficient drift for every selected sequence and onset at Bob. The chosen threshold and block length make both terms in \eqref{eq:finite-key-explicit-error} vanish. The key length is $N_K\log_2J_K=O(K\log K)=O(\sqrt{n_K})$. Appendix~\ref{app:proof-finite-key} gives the order in which the constants are chosen and states each union bound explicitly.
\end{proof}

Theorem~\ref{thm:finite-key} gives a finite, sublinear sufficient key length. It does not prove that this order or its constant is optimal. General classical--quantum covert results also exhibit square-root-scale key resources under finite-alphabet assumptions \cite{BullockCQ2025}, but the theorem above is specialized to the bosonic, worst-message, worst-onset block-activity architecture.

\section{Numerical Validation}
\label{sec:numerics}

We report four finite-parameter checks selected to probe the numerically
delicate parts of the analysis without treating simulation as a proof of the
asymptotic theorems.  Unless stated otherwise, the channel parameters are
$\eta=0.8$, $\bar n_{\rm in}=0$, and $\bar n_W=1$, and all logarithms and
payloads are in nats.  The calculations use independent random-number streams
derived from master seed 20260806.  The complete CPU-only implementation also
records machine-readable data, convergence diagnostics, and unit tests.

\subsection{Local covert cost and exact allocation}

Figure~\ref{fig:local-geometry}(a) evaluates the analytic thermal cost and the
QPSK cost computed by truncated-Fock eigendecomposition.  Here
$c_W=0.01$, while the smallest resolved QPSK ratio is $0.009997$; all displayed
QPSK points pass positivity, trace, and two-cutoff convergence checks.  The
agreement is a finite-energy check of the common quadratic coefficient, not a
numerical proof of the $O(E^4)$ remainder.  Figure~\ref{fig:local-geometry}(b)
uses deterministic zero-sum perturbations of the uniform energy vector.  Each
profile is rescaled to satisfy the same exact Willie budget, so the decrease
below one directly displays the loss in Bob's accumulated linear information.

\begin{figure*}[!t]
\centering
\includegraphics[width=\textwidth]{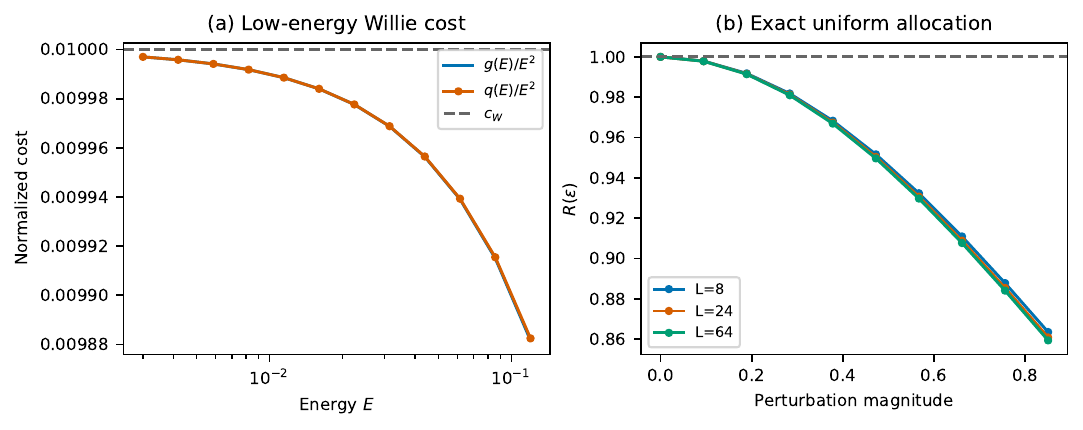}
\caption{Panel (a) varies the signal energy at $\eta=0.8$ and $\bar n_W=1$:
the exact analytic thermal cost and a converged truncated-Fock QPSK calculation
approach the common coefficient $c_W$.  Panel (b) deterministically perturbs
uniform allocations at fixed $\delta_b=0.02$ and rescales them to the exact QRE
budget; only the uniform allocation attains the normalized information
frontier.}
\label{fig:local-geometry}
\end{figure*}

\subsection{One-block converse and CUSUM achievability}

Figure~\ref{fig:latency-cusum} compares finite-block Monte Carlo behavior with
the detector-independent converse and the analytic CUSUM guarantees.  The
simulation uses the isotropic receiver with $\kappa=1.0667$.  Every displayed
empirical point is based on 20,000 trials per hypothesis and onset.  No
empirical latency fell below the converse.  Panel (a) is used only to identify
when both empirical block-error probabilities enter the target-error regime.

\begin{figure*}[!t]
\centering
\includegraphics[width=\textwidth]{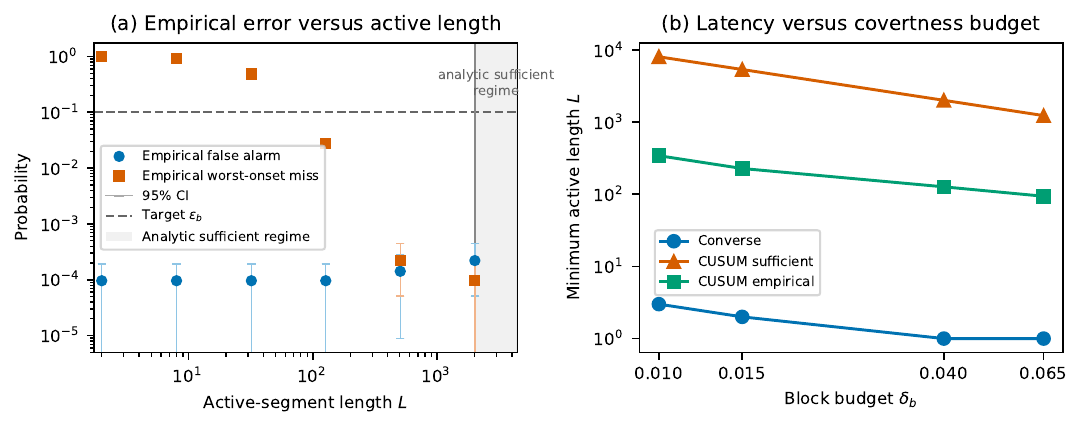}
\caption{One-block CUSUM validation for the isotropic Gaussian receiver at
$\eta=0.8$, target errors $\alpha_b=\beta_b=\epsilon_b=0.1$, and onset allowance
$L_{\rm start}=\lfloor0.25L\rfloor$.  Panel (a) compares empirical CUSUM
false-alarm and worst-onset missed-detection probabilities as functions of the
active length $L$ at fixed $\delta_b=0.04$.  Each point uses 20,000 Monte Carlo
trials per hypothesis and onset, and the error bars are Wilson 95\% confidence
intervals.  The horizontal line marks the target $\epsilon_b$, while the shaded
region beginning at $L_{\rm suff}=2011$ marks the analytic sufficient regime
from Theorem~\ref{thm:block-achievability}.  For each $L$, the energy is reset to
$E_L^*=g^{-1}(\delta_b/L)$, so the total block covertness budget remains fixed.
Zero-event estimates are placed at half their Wilson upper endpoint solely for
display on the logarithmic axis.  Panel (a) is a finite-parameter sanity check
of when CUSUM enters the target-error regime, not a numerical proof of the
asymptotic theorem.  Panel (b) varies $\delta_b$: ``Converse'' is the smallest
integer $L$ allowed by binary data processing for any block-deadline detector,
``CUSUM sufficient'' is the smallest $L$ for which one threshold satisfies all
analytic sufficient conditions jointly, and ``CUSUM empirical'' is the
smallest tested $L$ whose upper confidence limits meet both error targets.  The
three curves have different operational meanings and do not represent
equal-status receiver designs.}
\label{fig:latency-cusum}
\end{figure*}

\subsection{Block-activity payload and full-horizon reference}

Figure~\ref{fig:payload-reference} evaluates the constructive parameter
sequence with $N=\lceil2K\rceil$, $L=\lceil2500K(\log K)^2\rceil$, and
$B=L+\lfloor0.25L\rfloor$.  All sampled points satisfy the construction's
drift and error-exponent inequalities.  The normalization in panel (b) makes
the two theorem-implied orders visible without fitting a finite-range slope.

\begin{figure*}[!t]
\centering
\includegraphics[width=\textwidth]{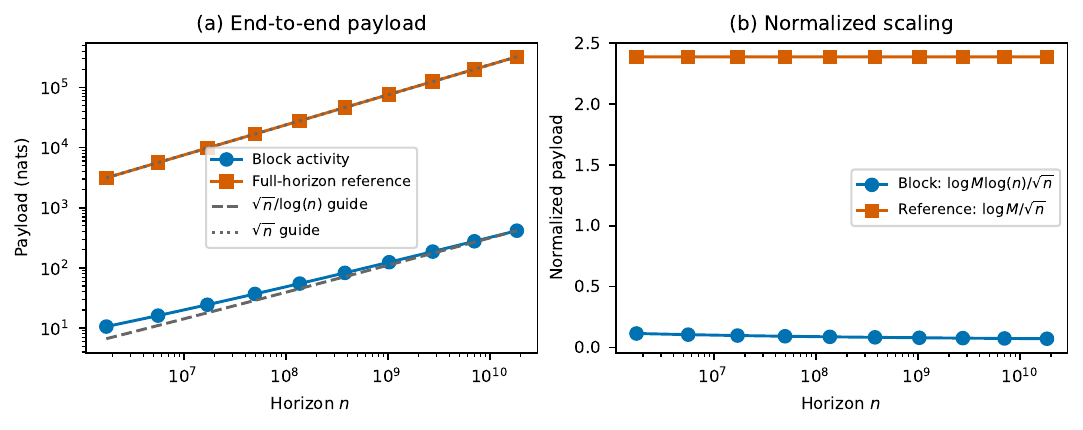}
\caption{Panel (a) uses the theorem's constructive block-activity parameter
sequence and compares it with the relaxed, receiver-restricted full-horizon
first-order reference; the light guides show $\sqrt n/\log n$ and $\sqrt n$.
Panel (b) displays the corresponding normalizations.  The block-activity curve
requires recovery by local deadlines, whereas the full-horizon reference
allows one joint decision after all $n$ observations.  The curves therefore
answer different operational questions and do not assert unrestricted covert
capacity or a ranking of competing designs.}
\label{fig:payload-reference}
\end{figure*}

\subsection{Finite-key QPSK sanity check}

Figure~\ref{fig:finite-key} checks the collision soft-covering bound and the
group-symmetry identity in finite Fock truncations.  For product length two,
the two largest cutoffs differ by $0.694\%$ in the displayed divergence.  The
empirical mean remains below the finite-dimensional collision bound, computed
with the actual truncated domination constant, and the Pythagorean residuals
remain at floating-point scale.

\begin{figure*}[!t]
\centering
\includegraphics[width=\textwidth]{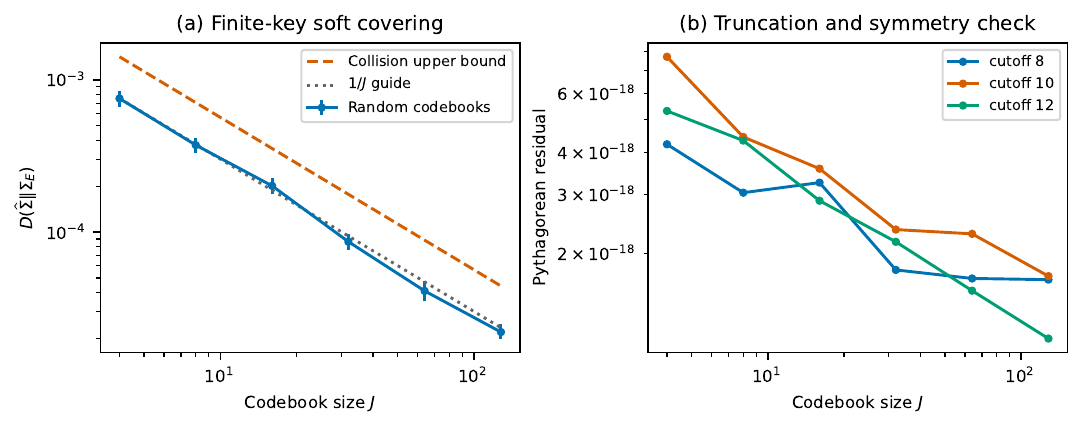}
\caption{Finite-dimensional numerical sanity check of the finite-key mechanism
at $E=0.01$ and product length two.  Panel (a) averages over 100 random public
codebooks for each $J$ and compares the empirical divergence, with 95\%
bootstrap intervals, against the collision upper bound computed using the
actual truncated domination constant.  The $1/J$ line is a visual guide, not
an asserted equality law.  Panel (b) shows the absolute Pythagorean residual
across Fock cutoffs 8, 10, and 12.  This calculation checks the small-scale
finite-dimensional mechanism; it is not a numerical proof of the
infinite-dimensional finite-key theorem.}
\label{fig:finite-key}
\end{figure*}

\FloatBarrier

\section{Conclusion}
\label{sec:conclusion}

This work develops and characterizes a covert block-activity information-transmission architecture in which the temporal activity vector is the message and must be recovered through block-local decisions by prescribed deadlines. The central result is a local-to-global transmission limit. At the one-block level, the exact thermal relative entropy at Willie and the linear information growth at Bob produce the finite-block energy--information frontier and the matching latency scaling $L^*(\epsilon_b,\delta_b)=\Theta[(c_W/(\kappa^2\delta_b))\log^2(1/\epsilon_b)]$. The converse is detector independent within the fixed measured channel, and a block-reset cumulative-sum detector attains the same order.

When the block decisions jointly represent one message, a fixed total Willie budget must be divided across active coordinates and whole-pattern reliability forces the local false-alarm and miss probabilities to decrease with the number of decisions. Combining these effects with the weight-constrained pattern count yields the end-to-end transmission law
\begin{equation}
\log M=\Theta\!\left(\frac{\sqrt n}{\log n}\right),
\end{equation}
with a matching converse within the symmetric shared-CSCG blockwise architecture. A relaxed full-horizon reference recovers $\Theta(\sqrt n)$ payload under the same modulation family and fixed Gaussian measurement. This comparison does not establish an unrestricted capacity separation; it shows that covertness and the selected measurement alone do not force the additional logarithmic factor appearing under the blockwise recovery architecture.

The information-theoretic novelty lies in this architectural coupling. Activity location is already an established information-bearing mechanism through pulse-position modulation and related sparse signaling constructions \cite{BlochGuha2017PPM,KadampotTahmasbiBloch2020PPM,TanAndersonBullockBash2026}. The contribution here is therefore not positional or support modulation by itself, but the information-theoretic coupling created when every coordinate of a message-bearing activity vector must be resolved locally under a shared total QRE budget and whole-pattern reliability. Theorems~\ref{thm:scaling-ach} and~\ref{thm:architecture-converse} quantify the resulting local-to-global transmission law.

The matching converse is deliberately architecture-specific. Assumption~\ref{ass:symmetric-architecture} uses the full weight-constrained pattern codebook, equal active-block covertness budgets, the shared-CSCG displaced-thermal signaling family of Section~\ref{sec:model}, a fixed Gaussian receiver, and identical block-local decision rules. It therefore rules out a larger payload order within that architecture, but does not preclude gains from selected subsets of activity patterns, unequal budget allocation, other bosonic input-state families, message-dependent signaling laws, or decisions coupling observations across blocks. Likewise, the constructive regime $N=\Theta(K)$ is weight constrained rather than asymptotically sparse in the distinct vanishing-density regime $K=o(N)$.

Finally, the finite-key construction removes the ideal continuous shared-randomness assumption: public quadrature phase-shift keying codebooks selected with $O(\sqrt n)$ secret bits preserve the same block-activity payload order. This is a sufficient key length, not a key-length converse. Willie is modeled through the lost-light weak-complementary output, and Bob's main results begin after a Gaussian measurement has been fixed; no optimality over collective quantum measurements is claimed. Extending the transmission converse beyond the present shared-CSCG block-local architecture, treating stronger access for Willie or broader quantum receivers at Bob, characterizing the $K=o(N)$ regime, and determining the minimum secret-key order remain natural open problems.


\clearpage
\appendices
\section{Deferred Proofs for Sequential Detection and Architecture Scaling}
\label{app:sequential-and-architecture}

This appendix gives the full concentration argument for the missed-detection bound and the complete converse for the block-by-block architecture.

\subsection{Proof of Lemma~\ref{lem:miss}}
\label{app:proof-miss}

We first prove the missed-detection bound uniformly over every allowed onset.

\begin{proof}
For an onset $\nu$, let
\begin{equation}
A_{L,\nu}=E_L^*\sum_{j=1}^LZ_{\nu+j}
\end{equation}
be the conditional KL accumulated during the active segment. On the event
\begin{equation}
\cG_{L,\nu}=\left\{L^{-1}\sum_{j=1}^LZ_{\nu+j}\ge(1-\xi)\kappa\right\},
\end{equation}
condition \eqref{eq:drift-condition} implies $A_{L,\nu}\ge(1+\gamma)h$. Conditioned on the secret symbols, the ordinary LLR sum over the $L$ active samples obeys
\begin{equation}
\Lambda_{L,\nu}:=\sum_{j=1}^L\ell_{\nu+j}(E_L^*)\sim\NN(A_{L,\nu},2A_{L,\nu})
\end{equation}
by \eqref{eq:llr-normal}. Hence
\begin{align}
\Prb_\nu(\Lambda_{L,\nu}<h\mid S^B,\cG_{L,\nu})
&\le\exp\left[-\frac{(A_{L,\nu}-h)^2}{4A_{L,\nu}}\right]\nonumber\\
&\le\exp\left[-\frac{\gamma^2}{4(1+\gamma)}h\right].
\end{align}
The complement of $\cG_{L,\nu}$ has probability at most $e^{-LI_-(\xi)}$ by \eqref{eq:key-lower-tail}. The CUSUM recursion implies inductively that $W_{\nu+j}\ge W_\nu+\sum_{i=1}^j\ell_{\nu+i}(E_L^*)$ for every $j$. Since $W_\nu\ge0$, the event $\Lambda_{L,\nu}\ge h$ forces $T_h\le\nu+L$. Combining these facts proves the result. Because the symbols are i.i.d., the bound is identical for every admissible onset.
\end{proof}

\subsection{Proof of Theorem~\ref{thm:architecture-converse}}
\label{app:proof-architecture-converse}

We next turn the per-block error requirements into a lower bound on block length and then into an upper bound on the number of messages.

\begin{proof}
If $N$ remains bounded along a subsequence, then $M\le2^N$ is bounded and the conclusion is immediate. Hence consider a subsequence on which $N\to\infty$. Let $\alpha_n$ be the common idle-block false-alarm probability and let
\begin{equation}
\beta_n:=\max_{0\le\nu\le L_{\rm start}}
\Prb_\nu(\widehat X=0)
\end{equation}
be the worst-onset active-block miss probability. Under the all-zero codeword, the block-local decisions are independent, and therefore
\begin{equation}
1-(1-\alpha_n)^N\le\epsilon_n.
\end{equation}
Consequently,
\begin{equation}
\alpha_n\le1-(1-\epsilon_n)^{1/N}
\le\frac{-\log(1-\epsilon_n)}{N}
=O(\epsilon_n/N).
\label{eq:alpha-necessary}
\end{equation}
Choose a weight-$K$ codeword and assign every active block an onset attaining the finite maximum in the definition of $\beta_n$. Independence across blocks and the maximal-error requirement imply
\begin{equation}
1-(1-\beta_n)^K\le\epsilon_n,
\end{equation}
and hence
\begin{equation}
\beta_n\le1-(1-\epsilon_n)^{1/K}
\le\frac{-\log(1-\epsilon_n)}{K}
=O(\epsilon_n/K).
\label{eq:beta-necessary}
\end{equation}

For $p=1-\beta_n$ and $q=\alpha_n$,
\begin{equation}
\db(p\|q)
\ge(1-\beta_n)\log\frac1{\alpha_n}-h_2(\beta_n)
\ge(1-o(1))\log N.
\label{eq:binary-kl-logN}
\end{equation}
Apply Theorem \ref{thm:block-converse} to one active block with budget $\delta/K$. Since the right-hand side of \eqref{eq:binary-kl-logN} diverges, the exact converse forces $KL\to\infty$, and therefore $u_n:=\delta/(KL)\to0$. Using \eqref{eq:g-inverse-expansion},
\begin{align}
\kappa Lg^{-1}(u_n)
&=\kappa\sqrt{\frac{\delta L}{c_WK}}+O(1/K).
\end{align}
Combining this identity with \eqref{eq:binary-kl-logN} yields
\begin{equation}
L\ge
\frac{c_W}{\kappa^2\delta}
K(\log N)^2(1-o(1)).
\label{eq:L-architecture-lower}
\end{equation}
Since $n=NB\ge NL$, there is a constant $c>0$ such that
\begin{equation}
n\ge cNK(\log N)^2
\label{eq:n-NK}
\end{equation}
for all sufficiently large $n$.

Set $r=N/K\ge2$. The standard binomial-sum bound gives
\begin{equation}
M=\sum_{j=0}^K\binom Nj
\le\left(\frac{eN}{K}\right)^K=(er)^K.
\label{eq:binomial-upper}
\end{equation}
From \eqref{eq:n-NK},
\begin{equation}
K\le\frac{\sqrt n}{\sqrt{cr}\log N}.
\end{equation}
If $N\le\sqrt n/\log n$, then $\log M\le N\log2=O(\sqrt n/\log n)$. Otherwise $\log N=\Omega(\log n)$, and
\begin{equation}
\log M
\le
\frac{\sqrt n}{\sqrt c\log N}
\frac{\log(er)}{\sqrt r}
=O\left(\frac{\sqrt n}{\log n}\right),
\end{equation}
because $\sup_{r\ge2}\log(er)/\sqrt r<\infty$.
\end{proof}

\section{Full-Horizon Relaxation: Model, Converse, and Achievability}
\label{app:benchmark-proofs}

This appendix analyzes the deliberately relaxed reference problem summarized
in Section~\ref{sec:global-benchmark}. Unlike the block-activity task in the
main text, this problem has no unknown change onset, no intermediate
active-or-idle decisions, and no detection-delay criterion. Bob waits until all
$n$ channel uses have been observed and then decodes a single message.

\subsection{Relaxed full-horizon model}

For message $m$, let the complex displacement amplitude at use $t$ be
\begin{equation}
 \alpha_{m,t}=\sqrt{E_{m,t}}\,u_{m,t}S_t,
 \qquad |u_{m,t}|=1,
 \label{eq:global-mod-family}
\end{equation}
where $S_t\sim\CN(0,1)$ are independent shared symbols and $u_{m,t}$ is
deterministic and message dependent. Thus the message is represented by a
length-$n$ phase sequence rather than by a vector of blockwise activity
indicators. Bob applies the same Gaussian measurement $(C,\Sigma_0)$ on every
use and jointly decodes from $(S^n,Y^n)$ after the full horizon. The exact
Willie cost of message $m$ is $\sum_{t=1}^n g(E_{m,t})$.

This family permits nonuniform energies and coding across any block boundaries
used by the real-time architecture. It remains restricted to multiplicative
CSCG modulation followed by the chosen fixed Gaussian measurement; it does not
cover nonlinear functions of the shared randomness, non-Gaussian input states,
or collective quantum measurements performed before that measurement.

\subsection{Converse for the relaxed model}
\label{app:proof-global-fixed-converse}

\begin{theorem}[Converse within the multiplicative-CSCG family]
\label{thm:global-fixed-converse}
Suppose every message satisfies
\begin{equation}
 \sum_{t=1}^n g(E_{m,t})\le\delta,
 \label{eq:global-per-message-covertness}
\end{equation}
and the average decoding error is at most $\epsilon<1$. Then
\begin{equation}
 (1-\epsilon)\log M-h_2(\epsilon)
 \le \kappa n g^{-1}\!\left(\frac{\delta}{n}\right).
 \label{eq:global-converse-main}
\end{equation}
Consequently,
\begin{equation}
 \log M\le \frac{\kappa}{1-\epsilon}
 \sqrt{\frac{\delta n}{c_W}}+O(1).
 \label{eq:global-converse-asymptotic}
\end{equation}
\end{theorem}

\begin{proof}
Let $P_m$ be the law of $(S^n,Y^n)$ under message $m$ and let $P_0$ be
the all-idle law. Write $R(u)\in\mathbb R^{2\times2}$ for the orthogonal
rotation induced by multiplication by a unit-modulus complex number $u$.
Conditioned on $X_t=\operatorname{vec}(S_t)$, the single-use divergence is
\begin{equation}
 \frac{\eta E_{m,t}}{2}
 X_t^\top R(u_{m,t})^\top G R(u_{m,t})X_t.
\end{equation}
Averaging over $X_t\sim\NN(0,I_2)$ gives
\begin{align}
 \frac{\eta E_{m,t}}{2}
 \Tr\!\left[R(u_{m,t})^\top G R(u_{m,t})\right]
 &=\frac{\eta E_{m,t}}{2}\Tr G\nonumber\\
 &=\kappa E_{m,t},
\end{align}
where orthogonal similarity preserves trace. Independence across uses
therefore yields
\begin{equation}
 D(P_m\|P_0)=\kappa\sum_{t=1}^nE_{m,t}.
\label{eq:global-message-divergence}
\end{equation}
Since $g$ is increasing and convex, Jensen's inequality and
\eqref{eq:global-per-message-covertness} imply
\begin{equation}
 g\!\left(\frac1n\sum_{t=1}^nE_{m,t}\right)
 \le \frac1n\sum_{t=1}^ng(E_{m,t})
 \le \frac{\delta}{n},
\end{equation}
and hence
\begin{equation}
 \sum_{t=1}^nE_{m,t}\le n g^{-1}(\delta/n).
\label{eq:global-energy-bound}
\end{equation}
For a uniform message and $\bar P=M^{-1}\sum_mP_m$, the
relative-entropy mixture identity gives
\begin{equation}
 \frac1M\sum_{m=1}^M D(P_m\|P_0)
 =I(M;S^n,Y^n)+D(\bar P\|P_0).
\end{equation}
The second term is nonnegative, while Fano's inequality gives
\begin{equation}
 I(M;S^n,Y^n)\ge(1-\epsilon)\log M-h_2(\epsilon).
\end{equation}
Combining these inequalities with
\eqref{eq:global-message-divergence}--\eqref{eq:global-energy-bound}
proves \eqref{eq:global-converse-main}. Finally,
\eqref{eq:g-inverse-expansion} gives
\begin{equation}
 n g^{-1}(\delta/n)
 =\sqrt{\frac{\delta n}{c_W}}+O(1),
\end{equation}
which proves \eqref{eq:global-converse-asymptotic}.
\end{proof}

\subsection{Dense antipodal achievability}

\begin{theorem}[Dense antipodal achievability]
\label{thm:dense-achievability}
For every fixed $\delta>0$, there exist full-horizon codes with exact
maximal-message Willie QRE $\delta$, vanishing maximal error, and
\begin{equation}
 \log M_n
 =\kappa n g^{-1}(\delta/n)(1-o(1))
 =\kappa\sqrt{\frac{\delta n}{c_W}}(1-o(1)).
 \label{eq:dense-achievability}
\end{equation}
\end{theorem}

We first record the two auxiliary facts used in its proof.

\subsubsection{Information-density threshold bound}

\begin{lemma}[Information-density threshold bound]
\label{lem:threshold-coding}
For a memoryless channel with input distribution $P_X$, define
\begin{equation}
 \imath(X^n;Y^n)
 :=\log\frac{dP_{Y^n|X^n}(Y^n|X^n)}{dP_{Y^n}(Y^n)}.
\end{equation}
For every integer $M\ge2$ and every $\gamma>0$, there exists an $M$-code
whose average error probability satisfies
\begin{equation}
 P_{\rm e}^{\rm av}
 \le \Prb\!\left[\imath(X^n;Y^n)\le \log M+\gamma\right]+e^{-\gamma}.
 \label{eq:threshold-coding-bound}
\end{equation}
\end{lemma}

\begin{proof}
This is Feinstein's information-density bound. Independently draw the
codewords from $P_X^{\otimes n}$ and use threshold decoding. Averaging the
union bound over the random codebook gives
\eqref{eq:threshold-coding-bound}; see, e.g.,
\cite{PolyanskiyWuITNotes}.
\end{proof}

\subsubsection{Low-energy binary information}
\label{app:proof-low-energy}

\begin{lemma}[Low-energy binary information]
\label{lem:low-energy-binary}
Set $E>0$ and let an equiprobable binary symbol $U\in\{-1,+1\}$ choose the
complex displacement amplitude $\alpha=\sqrt E\,US$. For the induced
memoryless channel $U\mapsto(S,Y)$, let $I_E$ be the mutual information and
$V_E$ the variance of the information density. Then
\begin{equation}
 I_E=\kappa E+O(E^2),\qquad V_E=O(E),\qquad E\downarrow0.
 \label{eq:low-energy-info}
\end{equation}
\end{lemma}

\begin{proof}
Conditioned on $X=\operatorname{vec}(S)$, whitening and matched filtering
reduce the sufficient statistic to
\begin{equation}
 Y_{\rm mf}=\sqrt d\,U+N_0,
 \qquad d=2ZE,
 \qquad N_0\sim\NN(0,1).
\end{equation}
Let $J(d)=I(U;\sqrt d\,U+N_0)$. The
mutual-information--minimum-mean-square-error (I--MMSE) identity implies
$J(d)=d/2+O(d^2)$ as $d\downarrow0$ \cite{GuoShamaiVerdu2005}. This local
expansion can be made uniform for the random value $d=2ZE$: for
$0\le d\le d_0$ use the Taylor remainder, while for $d>d_0$ use
$0\le J(d)\le\log2$ and $d/2\le d^2/(2d_0)$. Thus a finite constant $C$
exists such that
\begin{equation}
 |J(d)-d/2|\le C d^2,\qquad d\ge0.
\end{equation}
Since the Gaussian quadratic form $Z$ has a finite second moment,
\begin{equation}
 I_E=\E[J(2ZE)]=\kappa E+O(E^2).
\end{equation}
Under $U=+1$, the conditional information density is
$T-\log\cosh T$ with $T=d+\sqrt d\,N_0$. For $T\ge0$,
$0\le T-\log\cosh T\le T$, whereas for $T<0$ its absolute value is at
most $2|T|$. Hence
\begin{equation}
 |T-\log\cosh T|^2\le4T^2,
\end{equation}
and the conditional second moment is at most $4(d+d^2)$. Averaging over
$Z$ gives $V_E=O(E)$.
\end{proof}

\subsubsection{Proof of Theorem~\ref{thm:dense-achievability}}
\label{app:proof-dense-achievability}

\begin{proof}
Set $E_n=g^{-1}(\delta/n)$ and transmit
$\alpha_{m,t}=\sqrt{E_n}\,b_{m,t}S_t$ with binary codewords
$b_m^n\in\{-1,+1\}^n$. Circular symmetry makes Willie's averaged state
independent of the codeword, with QRE $ng(E_n)=\delta$ for every message.

Apply Lemma~\ref{lem:threshold-coding} to the binary-input channel in
Lemma~\ref{lem:low-energy-binary}. Choose a sequence $\zeta_n\downarrow0$
such that $\zeta_n^2\sqrt n\to\infty$, for example
$\zeta_n=n^{-1/8}$, and set
\begin{equation}
 \widetilde M_n
 =\left\lfloor\exp\!\left[(1-\zeta_n)nI_{E_n}\right]\right\rfloor,
 \qquad
 \gamma_n=\frac{\zeta_n}{2}nI_{E_n}.
\end{equation}
Because $nI_{E_n}=\Theta(\sqrt n)$,
\begin{equation}
 \log\widetilde M_n+\gamma_n
 \le(1-\zeta_n/2)nI_{E_n}
\end{equation}
for all sufficiently large $n$. Chebyshev's inequality therefore gives
\begin{equation}
 \Prb\!\left[
 \sum_{t=1}^n\imath_t
 \le(1-\zeta_n/2)nI_{E_n}
 \right]
 \le\frac{4V_{E_n}}{\zeta_n^2nI_{E_n}^2}=o(1),
\end{equation}
while $e^{-\gamma_n}=o(1)$. Lemma~\ref{lem:threshold-coding} thus yields
a deterministic codebook of size $\widetilde M_n$ with average error
$o(1)$. Removing the worse half of its codewords gives a subcode with
maximal error $o(1)$ and size $M_n\ge\widetilde M_n/2$. This expurgation
preserves exact per-message covertness and changes $\log M_n$ by only
$O(1)$.

Finally, Lemma~\ref{lem:low-energy-binary} and
$E_n=g^{-1}(\delta/n)$ give
\begin{align}
 nI_{E_n}
 &=\kappa nE_n+O(nE_n^2)\nonumber\\
 &=\kappa n g^{-1}(\delta/n)+O(1).
\end{align}
Together with $\zeta_n\to0$, this proves
\eqref{eq:dense-achievability}.
\end{proof}

For this construction, maximum-likelihood decoding is equivalently a
full-horizon correlation test: up to terms independent of $m$, Bob maximizes
\begin{equation}
 \sum_{t=1}^n b_{m,t}X_t^\top C^\top\Sigma_0^{-1}Y_t
\end{equation}
over the codebook. This decoder waits for the full horizon and does not define
an online change-detection stopping time.

\section{Key-Controlled Homodyne as a Receiver Instantiation}
\label{app:key-controlled-homodyne}

This appendix gives a concrete receiver instantiation for the fixed-measurement model in Section~\ref{sec:model}. Bob uses the current shared symbol to choose the homodyne local-oscillator direction before each measurement. The result below optimizes only within the class of single-quadrature homodyne policies; it is included to show how the receiver-dependent slope $\kappa$ can be specialized, not to claim optimality over general Gaussian or quantum measurements. Let the premeasurement quadrature covariance be $V\succ0$, with vacuum covariance $I_2/2$. We assume that $V$ is independent of the key and is the same under the active and idle hypotheses.

\begin{proposition}[Optimal direction within key-controlled homodyne]
\label{prop:adaptive-homodyne}
Suppose $X\sim\NN(0,I_2)$ is the current displacement-direction vector. Bob knows $X$ before measuring the current mode and may choose the homodyne direction $w$ as a measurable function of $X$. The conditional KL slope is
\begin{equation}
 Z_{\rm hom}(w,X)=\frac{\eta}{2}
 \frac{(w^\top X)^2}{w^\top Vw}.
 \label{eq:homodyne-slope}
\end{equation}
Among single-quadrature homodyne measurements it is maximized, for $X\ne0$, by
\begin{equation}
 w^*(X)=\frac{V^{-1}X}{\norm{V^{-1}X}},
 \label{eq:adaptive-direction}
\end{equation}
with conditional and average slopes
\begin{align}
 Z_{\rm hom}^{\rm ad}(X)&=\frac{\eta}{2}X^\top V^{-1}X,
 \label{eq:adaptive-conditional-slope}\\
 \kappa_{\rm hom}^{\rm ad}&=\frac{\eta}{2}\Tr(V^{-1}).
 \label{eq:adaptive-mean-slope}
\end{align}
\end{proposition}

\begin{proof}
The homodyne outcome has variance $w^\top Vw$ and active mean $\sqrt{\eta E}\,w^\top X$, so the equal-variance Gaussian KL is $E Z_{\rm hom}(w,X)$. Generalized Cauchy--Schwarz gives
\begin{equation}
 (w^\top X)^2\le(w^\top Vw)(X^\top V^{-1}X),
\end{equation}
with equality precisely when $w\propto V^{-1}X$. Averaging $X^\top V^{-1}X$ over $X\sim\NN(0,I_2)$ gives \eqref{eq:adaptive-mean-slope}.
\end{proof}

For a fixed homodyne direction, with $\lambda_{\min}(V)$ denoting the smallest eigenvalue of $V$,
\begin{equation}
 \kappa_{\rm hom}^{\rm fix,*}
 =\frac{\eta}{2\lambda_{\min}(V)}.
\end{equation}
Ideal heterodyne adds vacuum covariance $I_2/2$ and has
\begin{equation}
 \kappa_{\rm het}
 =\frac{\eta}{2}\Tr\!\left(V+\frac12I_2\right)^{-1}.
 \label{eq:heterodyne-slope}
\end{equation}
Because $V^{-1}-(V+I_2/2)^{-1}\succ0$,
\begin{equation}
 \kappa_{\rm hom}^{\rm ad}>\kappa_{\rm het}.
 \label{eq:adaptive-beats-heterodyne}
\end{equation}
For $V=vI_2$ with $v\ge1/2$,
\begin{equation}
 \kappa_{\rm hom}^{\rm fix,*}=\frac{\eta}{2v},\quad
 \kappa_{\rm hom}^{\rm ad}=\frac{\eta}{v},\quad
 \kappa_{\rm het}=\frac{\eta}{v+1/2}.
\end{equation}
Thus heterodyne is no worse than fixed homodyne in the isotropic case, with equality only at vacuum noise, whereas key-controlled homodyne is strictly better. This is an optimization only over single-quadrature homodyne policies; it is not an optimality statement over all Gaussian positive operator-valued measures (POVMs) or all quantum receivers. Optimized homodyne measurements for displaced Gaussian models have related precedents \cite{Aspachs2009,Roberson2021}.

Any conclusion that depends only on the linear-information identity \eqref{eq:joint-bob-kl} may therefore be specialized by replacing $\kappa$ with \eqref{eq:adaptive-mean-slope}. Results that additionally use the distribution or concentration of the conditional slope require the corresponding case-specific check. For QPSK and isotropic $V$, the conditional slope is deterministic; for anisotropic $V$, it takes finitely many bounded values and can be controlled by Hoeffding's inequality.

\section{Weighted Analyticity for the QPSK Expansions}
\label{app:qpsk-regularity}

The finite-key proof uses Taylor expansions of infinite-dimensional displaced thermal states. This appendix justifies those expansions in a weighted Hilbert--Schmidt norm and then applies them to the QPSK remainder and collision coefficient.

This appendix supplies the infinite-dimensional regularity used in
Lemmas~\ref{lem:qpsk-cost} and \ref{lem:collision-expansion}. Put
\begin{align}
 \tau&=\tau_{N_W}=\sum_{n\ge0}p_n\ket n\bra n,\nonumber\\
 p_n&=(1-r)r^n,
 \qquad r=\frac{N_W}{N_W+1}\in(0,1).
\end{align}
The weighted norm $\norm{\cdot}_{2,\tau}$ and the notation $O_{2,\tau}$ are defined in \eqref{eq:weighted-hs-norm} and the paragraph following it. We also write $\operatorname{ad}_K(A):=[K,A]$ and $\{A,C\}:=AC+CA$. For $s\in\{1,i,-1,-i\}$, let
$K_s=s a^\dagger-s^*a$, $\varepsilon=\sqrt{a_WE}$, and
\begin{equation}
 \rho_s(\varepsilon)=e^{\varepsilon K_s}\tau e^{-\varepsilon K_s}
 =e^{\varepsilon\operatorname{ad}_{K_s}}(\tau).
 \label{eq:appendix-rho-eps}
\end{equation}

\begin{lemma}[Uniform thermal domination]
\label{lem:thermal-domination}
For all sufficiently small $\varepsilon$, with
$\sigma_E=\frac14\sum_s\rho_s(\varepsilon)$ and
$\omega_E=\tau_{N_W+\varepsilon^2}$,
\begin{align}
 \sigma_E&\succeq e^{-\varepsilon^2/N_W}\tau,
 \label{eq:sigma-lower-domination}\\
 \omega_E&\succeq
 \frac{N_W+1}{N_W+\varepsilon^2+1}\,\tau.
 \label{eq:omega-lower-domination}
\end{align}
Consequently the inverse-quarter-power sandwiches used below are well defined,
and their weighted Hilbert--Schmidt norms are uniformly controlled by the
corresponding $\tau$-weighted norm in a neighborhood of zero.
\end{lemma}

\begin{proof}
The thermal state has positive Glauber--Sudarshan density $P_\tau(z)=(\pi N_W)^{-1}e^{-|z|^2/N_W}$. The positive $P$ density of the QPSK average is
\begin{align}
 P_{\sigma_E}(z)
 &=P_\tau(z)e^{-\varepsilon^2/N_W}
 \frac14\sum_{s\in\{1,i,-1,-i\}}
 \exp\!\left(\frac{2\varepsilon\Re(zs^*)}{N_W}\right).
\end{align}
The average of the four real exponents is zero, so Jensen's inequality gives
$P_{\sigma_E}(z)\ge e^{-\varepsilon^2/N_W}P_\tau(z)$ pointwise. Integrating
against the positive coherent-state projectors proves
\eqref{eq:sigma-lower-domination}. For the thermal state,
$p_n(N_W+\varepsilon^2)/p_n(N_W)$ is increasing in $n$, and its minimum at
$n=0$ is $(N_W+1)/(N_W+\varepsilon^2+1)$, proving
\eqref{eq:omega-lower-domination}.

If $\chi\succeq c\tau$ with $c>0$, operator monotonicity of the
square root gives $\tau^{1/2}\preceq c^{-1/2}\chi^{1/2}$ as quadratic
forms. Hence the densely defined operator
$R_\chi=\chi^{-1/4}\tau^{1/4}$ extends to a bounded operator satisfying
$R_\chi R_\chi^\dagger\preceq c^{-1/2}\1$ and therefore
$\|R_\chi\|\le c^{-1/4}$. For every self-adjoint $A$ with finite
$\|A\|_{2,\tau}$,
\begin{align}
 \chi^{-1/4}A\chi^{-1/4}
 &=R_\chi(\tau^{-1/4}A\tau^{-1/4})R_\chi^\dagger,
\end{align}
first on the common finite-particle core (the finite linear span of Fock states) and then by Hilbert--Schmidt
closure. The ideal property of the Hilbert--Schmidt class now gives
\begin{equation}
 \|\chi^{-1/4}A\chi^{-1/4}\|_2
 \le c^{-1/2}\|\tau^{-1/4}A\tau^{-1/4}\|_2.
 \label{eq:weighted-domination-bound}
\end{equation}
This proves the stated uniform control without an interchange of compressed
inverse operators.
\end{proof}

\begin{lemma}[Uniform weighted displacement expansion]
\label{lem:weighted-displacement}
There is an $\varepsilon_0>0$ such that the maps
$\varepsilon\mapsto\rho_s(\varepsilon)$ are four times differentiable in
$\norm{\cdot}_{2,\tau}$ on $[-\varepsilon_0,\varepsilon_0]$, uniformly in the
four QPSK phases. Consequently,
\begin{align}
 \sigma_E&=\tau+\varepsilon^2\dot\tau_W+O_{2,\tau}(\varepsilon^4),\label{eq:appendix-sigma-exp}\\
 \tau_{N_W+\varepsilon^2}&=\tau+\varepsilon^2\dot\tau_W+O_{2,\tau}(\varepsilon^4),
 \label{eq:appendix-thermal-exp}
\end{align}
where $\dot\tau_W:=\partial_N\tau_N|_{N=N_W}$. In addition, with
$\omega_E=\tau_{N_W+a_WE}$,
\begin{equation}
 \norm{\omega_E^{-1/4}(\sigma_E-\omega_E)\omega_E^{-1/4}}_2
 \le C E^2
 \label{eq:appendix-moving-weight}
\end{equation}
for all sufficiently small $E$ and a finite constant $C$.
\end{lemma}

\begin{proof}
Because $N_W>0$, the displaced thermal state has the positive
Glauber--Sudarshan representation
\begin{equation}
 \rho_s(\varepsilon)=\int_{\mathbb C}\frac{d^2z}{\pi N_W}
 e^{-|z-\varepsilon s|^2/N_W}\ket z\bra z.
 \label{eq:appendix-p-representation}
\end{equation}
Using $\langle m|z\rangle=e^{-|z|^2/2}z^m/\sqrt{m!}$ gives
\begin{align}
 \bra m\rho_s(\varepsilon)\ket n
 &=\frac{1}{\pi N_W\sqrt{m!n!}}\nonumber\\
 &\quad\times\int_{\mathbb C}e^{-|z-\varepsilon s|^2/N_W-|z|^2}
 z^m(z^*)^n\,d^2z.
 \label{eq:appendix-fock-integral}
\end{align}
Fix $r_1$ with $r<r_1<\sqrt r$. Since
$1/r_1<1/r$, choose $c_1$ such that
$1/r_1<c_1<1/r$ and then decrease $\varepsilon_0$ so that completion of
the square bounds every derivative of order $j\le4$ in
\eqref{eq:appendix-fock-integral} by a fixed polynomial in $|z|$ times
$e^{-c_1|z|^2}$. Radial integration and Stirling's bounds give a constant
$C_j$ and a polynomial $P_j$ for which
\begin{equation}
 \left|\partial_\varepsilon^j
 \bra m\rho_s(\varepsilon)\ket n\right|
 \le C_j P_j(m+n) r_1^{(m+n)/2},
 \qquad |\varepsilon|\le\varepsilon_0.
 \label{eq:appendix-matrix-majorant}
\end{equation}
Indeed, the radial integral is proportional to
$c_1^{-(m+n+j+2)/2}\Gamma[(m+n+j+2)/2]/\sqrt{m!n!}$, and the gamma-to-factorial
ratio is polynomially bounded uniformly in $m,n$. The estimate is uniform in
$s$. After multiplication by $p_m^{-1/4}p_n^{-1/4}$ and squaring, the resulting
double series is summable because $r_1/\sqrt r<1$. Dominated convergence
therefore permits termwise differentiation through fourth order in the
weighted Hilbert--Schmidt norm.

Averaging over QPSK gives $\E s=\E s^2=0$ and $\E|s|^2=1$, hence
\begin{align}
 \E_s\operatorname{ad}_{K_s}(\tau)&=0,\\
 \frac12\E_s\operatorname{ad}_{K_s}^2(\tau)
 &=a^\dagger\tau a+a\tau a^\dagger
 -\frac12\{2a^\dagger a+1,\tau\}.
 \label{eq:additive-noise-generator}
\end{align}
For the diagonal thermal probabilities, with the convention $p_{-1}=0$,
the $n$th diagonal entry of \eqref{eq:additive-noise-generator} is
\begin{align}
 &n p_{n-1}+(n+1)p_{n+1}-(2n+1)p_n\nonumber\\
 &\qquad=p_n\frac{n-N_W}{N_W(N_W+1)}
 =\left.\partial_Np_n(N)\right|_{N=N_W}.
\end{align}
Thus
$\frac12\E_s\operatorname{ad}_{K_s}^2(\tau)
=\partial_N\tau_N|_{N=N_W}$. QPSK phase averaging removes all odd powers of
$\varepsilon$, which proves \eqref{eq:appendix-sigma-exp}. Ordinary Taylor
expansion of the thermal family proves \eqref{eq:appendix-thermal-exp}.

Equations \eqref{eq:appendix-sigma-exp} and
\eqref{eq:appendix-thermal-exp} give
$\|\sigma_E-\omega_E\|_{2,\tau}=O(\varepsilon^4)$.
The thermal domination \eqref{eq:omega-lower-domination} and
\eqref{eq:weighted-domination-bound} therefore imply directly that
\begin{equation}
 \|\omega_E^{-1/4}(\sigma_E-\omega_E)\omega_E^{-1/4}\|_2
 =O(\varepsilon^4)=O(E^2),
\end{equation}
which proves \eqref{eq:appendix-moving-weight}.

\end{proof}

The weighted expansion and domination estimates above are the only infinite-dimensional regularity inputs used in the deferred QPSK-cost and collision proofs.

\section{Deferred Proofs for the Finite-Key Construction}
\label{app:finite-key-proofs}

This appendix proves the QPSK approximation, the collision bound, and the final finite-key construction in the order in which they are used in the main text.

\subsection{Proof of Lemma~\ref{lem:qpsk-cost}}
\label{app:proof-qpsk-cost}

We first compare the QPSK average state with the thermal state of the same mean energy.

\begin{proof}
The QPSK average has mean photon number $N_W+a_WE$. Since
$\log\tau_N=-\log(N+1)+\hat n\log[N/(N+1)]$ is affine in the number
operator $\hat n=a^\dagger a$, the states $\sigma_E$ and
$\tau_{N_W+a_WE}$ have the same expectation of
$\log\tau_{N_W+a_WE}-\log\tau_{N_W}$. Therefore the exact projection
identity
\begin{equation}
 D(\sigma_E\|\tau_{N_W})
 =D(\sigma_E\|\tau_{N_W+a_WE})
 +D(\tau_{N_W+a_WE}\|\tau_{N_W})
 \label{eq:qpsk-projection}
\end{equation}
holds. QPSK cancels the odd displacement orders and matches circular Gaussian modulation through the order-$E$ state perturbation. In the thermal weighted Hilbert--Schmidt norm,
\begin{equation}
 \sigma_E-\tau_{N_W+a_WE}=O(E^2).
 \label{eq:qpsk-state-difference}
\end{equation}
Appendix~\ref{app:qpsk-regularity} proves the corresponding moving-reference bound
\begin{equation}
 \|\tau_{N_W+a_WE}^{-1/4}
 (\sigma_E-\tau_{N_W+a_WE})
 \tau_{N_W+a_WE}^{-1/4}\|_2=O(E^2).
\end{equation}
Monotonicity of the sandwiched R\'enyi divergence in its order then gives
\begin{equation}
 D(\sigma_E\|\tau_{N_W+a_WE})=O(E^4).
\end{equation}
The second term in \eqref{eq:qpsk-projection} is the thermal cost $g(E)=c_WE^2+O(E^3)$, proving \eqref{eq:qpsk-cost-expansion}. Since $c_W>0$, the same expansion gives $q(E)>0$ for all sufficiently small $E>0$. Details of the weighted expansion are recorded in Appendix~\ref{app:qpsk-regularity}.
\end{proof}

\subsection{Proof of Lemma~\ref{lem:collision-expansion}}
\label{app:proof-collision-expansion}

We next compute the first nonzero term of the fixed-reference collision quantity.

\begin{proof}
Put $\varepsilon=\sqrt{a_WE}$. The weighted expansion in
Appendix~\ref{app:qpsk-regularity} gives, uniformly over the four phases,
\begin{align}
 \rho_r(E)&=\tau_{N_W}+\varepsilon A_r+O_{2,\tau}(\varepsilon^2),\nonumber\\
 \sigma_E&=\tau_{N_W}+\varepsilon^2\dot\tau_W+O_{2,\tau}(\varepsilon^4),
\end{align}
where $K_r:=i^r a^\dagger-i^{-r}a$, $A_r=[K_r,\tau_{N_W}]$, $\Tr A_r=0$, and $\Tr\dot\tau_W=0$. The order-$\varepsilon^2$ coefficient of each trace-preserving state expansion also has zero trace. After sandwiching by the fixed thermal weight and expanding the square, the linear and second-order cross terms therefore vanish. The QPSK average is invariant under $\varepsilon\mapsto-\varepsilon$, so all remaining odd powers cancel. It follows that
\begin{equation}
 Q_{2,\tau}(E)=1+\frac{\varepsilon^2}{4}\sum_{r=0}^3
 \norm{\tau_{N_W}^{-1/4}A_r\tau_{N_W}^{-1/4}}_2^2
 +O(\varepsilon^4).
\end{equation}
All four norms are equal. If
$p_n=N_W^n/(N_W+1)^{n+1}$, then, for a real displacement,
\begin{align}
 \norm{\tau_{N_W}^{-1/4}[a^\dagger-a,\tau_{N_W}]
 \tau_{N_W}^{-1/4}}_2^2
 &=2\sum_{n\ge0}(n+1)
 \frac{(p_n-p_{n+1})^2}{\sqrt{p_np_{n+1}}}\nonumber\\
 &=\frac{2}{\sqrt{N_W(N_W+1)}}.
\end{align}
This proves \eqref{eq:collision-expansion}. For the average state, the
order-$\varepsilon^2$ cross term is $2\varepsilon^2\Tr\dot\tau_W=0$, which gives
\eqref{eq:average-collision-expansion}. Equation
\eqref{eq:collision-exponential} follows by continuity. Finally, expanding
the centered square and using $\E_r\rho_r(E)=\sigma_E$ gives the exact
variance identity \eqref{eq:qpsk-collision-variance}.
\end{proof}

\subsection{Proof of Lemma~\ref{lem:collision-soft-covering}}
\label{app:proof-soft-covering}

This step turns the collision quantity into a relative-entropy bound for a finite empirical mixture.

\begin{proof}
Write $\Delta=\widehat\sigma-\sigma$. Since $\Tr\Delta=0$, monotonicity of
the sandwiched R\'enyi divergence in its order and $\log(1+x)\le x$ give
\begin{align}
 D(\widehat\sigma\|\sigma)
 &\le\log\!\left(1+
 \norm{\sigma^{-1/4}\Delta\sigma^{-1/4}}_2^2\right)\nonumber\\
 &\le\norm{\sigma^{-1/4}\Delta\sigma^{-1/4}}_2^2.
\end{align}
The domination $\sigma\succeq c\tau$ and
\eqref{eq:weighted-domination-bound} imply
\begin{equation}
 D(\widehat\sigma\|\sigma)
 \le c^{-1}\norm{\tau^{-1/4}\Delta\tau^{-1/4}}_2^2.
\end{equation}
Expanding the centered square and using independence gives
\begin{equation}
 \E_{\cC}\norm{\tau^{-1/4}\Delta\tau^{-1/4}}_2^2
 =\frac{Q_{2,\tau}-R_{2,\tau}}{J}.
\end{equation}
This proves \eqref{eq:collision-soft-covering}. Tensor products give the
product statement. All quantities used here are finite in the QPSK
application by Lemmas~\ref{lem:thermal-domination} and
\ref{lem:collision-expansion}; the infinite-dimensional sandwiched R\'enyi
inequality is used only under this finiteness condition \cite{LiGaoHao2021}.
\end{proof}

\subsection{Proof of Theorem~\ref{thm:finite-key}}
\label{app:proof-finite-key}

Finally, we choose the parameters and one deterministic family of public codebooks that is good for both Willie and Bob.

\begin{proof}
Let $c_\gamma=\gamma^2/[4(1+\gamma)]$. First choose
$H>\max\{2,c_\gamma^{-1}\}$ and then choose $A$ so large that
\begin{equation}
 (1-\xi)\kappa\sqrt{\frac{\delta A}{c_W}}
 >(1+\gamma)H.
 \label{eq:finite-key-constant-order}
\end{equation}
Set $\theta_K=1/\log K$. For all sufficiently large $K$, so that
$2\theta_K<1$, use the explicit energy
\begin{equation}
 E_K:=\sqrt{\frac{(1-2\theta_K)\delta}{c_WK L_K}}.
 \label{eq:finite-key-energy}
\end{equation}
Because $E_K=\Theta((K\log K)^{-1})=o(\theta_K)$, Lemma~\ref{lem:qpsk-cost} gives
\begin{align}
 L_Kq(E_K)
 &= (1-2\theta_K)\frac{\delta}{K}[1+O(E_K)]
 \le (1-\theta_K)\frac{\delta}{K},
 \label{eq:finite-key-q-budget}\\
 L_KE_K&=\sqrt{\frac{(1-2\theta_K)\delta L_K}{c_WK}}
 =\Theta(\log K).
 \label{eq:finite-key-energy-asymptotic}
\end{align}
For each block, generate $J_K$ independent public length-$B_K$ QPSK sequences. A secret uniform index selects one sequence. For a fixed block and onset, averaging over the subkey produces an empirical mixture on the active length-$L_K$ substring. Lemma~\ref{lem:thermal-domination} gives
$\sigma_{E_K}\succeq e^{-a_WE_K/N_W}\tau_{N_W}$. Hence
Lemmas~\ref{lem:collision-expansion} and
\ref{lem:collision-soft-covering} imply
\begin{align}
 \E D(\widehat\Sigma_{b,\nu}\|\Sigma_{E_K})
 &\le\frac{\exp[(d_2+a_W/N_W)L_KE_K]}{J_K}\nonumber\\
 &=K^{C_0-r_{\rm key}+o(1)}
 \label{eq:finite-key-soft-covering-expectation}
\end{align}
for a finite constant $C_0$ (which may depend on the now-fixed design constants, including $A$, but not on $K$). Choose
$r_{\rm key}>C_0+4$. The explicit scaling in \eqref{eq:finite-key-scaling} gives
\begin{equation}
 N_K(L_{{\rm start},K}+1)=O(K^2(\log K)^2).
 \label{eq:finite-key-pair-count}
\end{equation}
Markov's inequality, with target $\theta_K\delta/K=\delta/(K\log K)$, and a union bound show that the probability that \emph{any} block--onset pair violates the target is at most
\begin{equation}
 K^{C_0-r_{\rm key}+3+o(1)}.
 \label{eq:finite-key-soft-failure}
\end{equation}
With the preceding choice of $r_{\rm key}$, this probability tends to zero and yields deterministic public codebooks satisfying, simultaneously for every block and onset,
\begin{equation}
 D(\widehat\Sigma_{b,\nu}\|\Sigma_{E_K})
 \le\theta_K\frac{\delta}{K}.
 \label{eq:finite-key-soft-covering-uniform}
\end{equation}
The active-substring Pythagorean identity gives
\begin{equation}
D(\widehat\Sigma_{b,\nu}\|\tau_{N_W}^{\otimes L_K})
\le \theta_K\frac{\delta}{K}+L_Kq(E_K)\le\frac{\delta}{K}.
\end{equation}
Tensoring the unchanged idle positions does not alter this divergence. Independent block subkeys make the message state a tensor product across blocks, so additivity yields \eqref{eq:finite-key-covertness} for every message of weight at most $K$.

For reliability, consider the random generation of the public QPSK codebooks before they are fixed. For a uniformly distributed QPSK symbol $S\in\{1,i,-1,-i\}$, the induced slope takes finitely many bounded values and obeys the exact identity
\begin{equation}
\E_{\rm QPSK} Z
=\frac{\eta}{2}\,\E_{\rm QPSK}[X^\top G X]
=\frac{\eta}{2}\Tr G
=\kappa,
\label{eq:qpsk-mean-slope}
\end{equation}
because the four phases have second moment $\E_{\rm QPSK}[XX^\top]=I_2$. Thus, for a fixed block, codeword index, and onset, Hoeffding's inequality with probability taken over the random public-codebook generation gives
\begin{equation}
\Prb_{\cC}\!\left[\frac1{L_K}\sum_{t=1}^{L_K}Z_t<(1-\xi)\kappa\right]
\le e^{-c_\xi L_K}
\label{eq:finite-key-hoeffding}
\end{equation}
for some $c_\xi>0$ (and the event is empty if the slope is deterministic). The total number of codeword--block--onset triples is at most
$N_KJ_K(L_{{\rm start},K}+1)=K^{r_{\rm key}+2+o(1)}$, whereas
$e^{-c_\xi L_K}=e^{-\Theta(K\log^2K)}$. A union bound over these triples therefore shows that the random-codebook probability of any drift violation tends to zero. Intersecting this event with the preceding soft-covering event yields, for all sufficiently large $K$, at least one deterministic family of public codebooks that is simultaneously good for Willie and has the required drift for every selected sequence and onset at Bob. After this deterministic family is fixed, the drift property is no longer probabilistic.

Conditional on any selected sequence, Bob runs the same block-reset CUSUM with the known-symbol likelihood increments $\ell_t(E_K)$. The false-alarm martingale argument of Lemma~\ref{lem:false-alarm} is parameter independent and is therefore unchanged. The deterministic drift lower bound removes the term $e^{-L_KI_-(\xi)}$ from Lemma~\ref{lem:miss}; with $c_\gamma=\gamma^2/[4(1+\gamma)]$,
\begin{equation}
P_{\rm e}^{\max}
\le N_KB_Ke^{-h_K}+K e^{-c_\gamma h_K}.
\label{eq:finite-key-explicit-error}
\end{equation}
The choices in \eqref{eq:finite-key-constant-order} and
\eqref{eq:finite-key-energy-asymptotic} ensure
$(1-\xi)\kappa L_KE_K\ge(1+\gamma)h_K$ for all sufficiently large $K$.
The selected value of $H$ also makes both terms in
\eqref{eq:finite-key-explicit-error} tend to zero. Finally, \eqref{eq:finite-key-length} and \eqref{eq:finite-key-payload} follow from \eqref{eq:finite-key-scaling}, \eqref{eq:finite-key-codebook-size}, and the weight-constrained pattern count.
\end{proof}

\emergencystretch=1em
\bibliographystyle{IEEEtran}
\bibliography{refs}

\end{document}